\documentclass[12pt,halfline,a4paper]{ouparticle}

\usepackage[utf8]{inputenc} 
\usepackage[T1]{fontenc}    
\usepackage{hyperref}       
\hypersetup{
    colorlinks = true,
    linkcolor  = blue,   
    citecolor  = blue,   
    urlcolor   = blue    
}

\usepackage{url}            
\usepackage{float}
\usepackage{nicefrac}       
\usepackage{microtype}      
\usepackage{lipsum}
\usepackage{fancyhdr}       
\usepackage{amsfonts, amssymb, amsmath, amsthm, mathtools, bm}
\usepackage{tikz}
\usetikzlibrary{calc, decorations.pathreplacing, arrows.meta, positioning}
\usepackage{bbm}
\usepackage{graphicx}       
\graphicspath{{media/}}     
\usepackage[dvipsnames]{xcolor}
\usepackage{dsfont}

\usepackage{algorithm}    
\usepackage{algpseudocode}

\usepackage[style=authoryear, backend=biber, sorting=nyt, maxbibnames=5, maxcitenames=2, natbib]{biblatex}
\bibliography{references.bib}

\newtheorem{theorem}{Theorem}  
\theoremstyle{remark}
\newtheorem{remark}{Remark}  
\theoremstyle{plain}  
\newtheorem{lemma}{Lemma}    

\begin{document}

\title{Fast and scalable inference in \\hidden Markov models with Gaussian fields}
\author{
    Jan-Ole Fischer\thanks{
      Email: \href{mailto:jan-ole.fischer@uni-bielefeld.de}{\texttt{jan-ole.fischer@uni-bielefeld.de}}.
      ORCID: \href{https://orcid.org/0009-0004-1556-9053}{0009-0004-1556-9053}
    }\hspace{.2cm}\\
    Department of Business Administration and Economics, Bielefeld University
}

\abstract{
Hidden Markov models (HMMs) are powerful tools for analysing time series data that depend on discrete 
underlying but unobserved states.
As such, they have gained prominence across numerous empirical disciplines, in particular ecology, medicine, and economics.
However, the increasing complexity of empirical data is often accompanied by additional latent structure
such as spatial effects, temporal trends, or measurement perturbations.
Gaussian fields provide an attractive building block for incorporating such structured latent variation into HMMs.

Fast inference methods for Gaussian fields have emerged through the stochastic partial differential equation (SPDE) approach. Due to their sparse representation, these integrate well with novel frequentist estimation methods for random-effects models via the use of automatic differentiation and the Laplace approximation.
Scaling to high dimensions requires tools such as \texttt{(R)TMB} to exploit sparsity in the Hessian w.r.t.\ the latent variables --- a property satisfied by SPDE fields but violated by the HMM likelihood.
We present a modified 
forward algorithm to compute the HMM likelihood, constructing sparsity in the Hessian and 
consequently enabling fast and scalable inference. 

We demonstrate the practical feasibility and the usefulness through simulations and two case studies exploring the detection of stellar flares 
as well as modelling the movement of lions.
}

\date{\today}

\keywords{Markov-switching models, state-space models, Gaussian process, Laplace approximation, automatic differentiation}

\maketitle

\section{Introduction}

Over the past decade, hidden Markov models (HMMs) have become a powerful tool for modelling noisy time series driven by unobserved, serially correlated states. Their versatility has led to widespread adoption across disciplines such as ecology \citep{mcclintock2020uncovering,mews2022multistate}, finance \citep{liu2012stock,zhang2019high}, medicine \citep{amoros2019continuous,soper2020hidden}, and sports \citep{otting2023copula, michels2026integrating}.
Recent advances have focused on replacing restrictive and potentially unrealistic parametric assumptions with more flexible model components. A notable development is the integration of penalised splines into HMMs, first proposed by \citet{langrock2015nonparametric}. Subsequent work \citep{langrock2017markov,chen2023bayesian,feldmann2023flexible,koslik2025flexible,ammann2026non} has expanded this approach and focussed on selecting an adequate degree of smoothness through a mixed model representation \citep{koslik2024efficient, michelot2025hmmtmb} as well as multivariate extensions \citep{koslik2025tensor,michels2025nonparametric}.

While penalised splines have significantly advanced the flexibility of HMMs beyond traditional parametric formulations, their low-rank function approximations involve inherent trade-offs. Current implementations often rely on calculations with dense matrices, which can become computationally demanding as the dimensionality of the problem increases. Additionally, the penalties in penalised splines are primarily designed to enforce global function smoothness, which may not always capture more complex or domain-specific patterns in the data.
On the other hand, Gaussian processes (GPs) \citep{ohagan1978curve} --- also referred to as Gaussian \emph{fields} (GFs) in spatial applications \citep{cressie2015statistics} --- offer a different, rich framework for function approximation.
Through the choice of covariance kernels, GPs can encode prior knowledge about the underlying structure of the function --- such as periodicity (e.g.\ \citealp{foreman2017fast}), spatial correlations, or non-stationarity --- making them particularly well-suited for applications where the data exhibits such intricate patterns. This adaptability opens up new possibilities for modelling complex phenomena that may not be fully captured by penalised splines alone. 
While historically GPs have been limited by their poor scalability, \citet{lindgren2011explicit} developed a method to approximate Matérn fields on a triangulated mesh, transforming the finite-dimensional approximation into a Gaussian Markov random field (GMRF). This approach enables sparse computations and has gained popularity through the \texttt{R-INLA} \citep{lindgren2015bayesian} and \texttt{inlabru} \citep{bachl2019inlabru} \texttt{R} packages, the latter tailored for distributional regression models with latent Gaussian fields.

Integrating such fields into HMMs can be beneficial in two distinct scenarios. 
First, in cases where the state-dependent observation process within the HMM is perturbed by e.g.\ GPS measurement error or quasi-periodic trend \citep{mcclintock2012general, esquivel2025detecting}, this can be accounted for.
Second, in more standard mixed model settings --- such as those considered in \texttt{inlabru} for regression --- a component of the HMM (e.g., transition probabilities between states) can be related to a latent (spatial) Gaussian field. 
Fast frequentist likelihood-based inference for both scenarios is desirable but challenging due to the high-dimensional integration required to evaluate the likelihood function. The Template Model Builder 
framework and its associated \texttt{R} packages, \texttt{TMB} \citep{kristensen2016tmb} and \texttt{RTMB} \citep{kristensen2024rtmb}, address this challenge by leveraging automatic differentiation to perform the Laplace approximation of the high-dimensional integral. However, the approach requires computing the Hessian matrix of the log-likelihood with respect to the latent variables, which is only feasible in high-dimensional settings if the Hessian is a sparse matrix. Unfortunately, such a sparsity pattern is not present in HMMs, necessitating a tailored approximation of the HMM log-likelihood function.

In this paper, we propose such an approximation by modifying the forward algorithm to compute the HMM log-likelihood. This modification ensures local time-dependence only and consequently yields a sparse second derivative with respect to latent variables spaced out over time. Combined with the SPDE approach by \citet{lindgren2011explicit}, this enables the integration of Gaussian fields into HMMs and facilitates rapid, numerically efficient inference that scales well to high dimensions.
We demonstrate the practical utility and performance of our approach through two case studies: (1) modelling stellar flares, where brightness measurements are perturbed by quasi-periodic oscillations in star brightness, and (2) movement ecology, where a spatial field influences the probability of transitions between behaviours.

\section{Methods}
\label{sec: methods}

Before we detail the novel inference procedures outlined above, we shall briefly introduce the two model components required throughout this paper.

\subsection{Hidden Markov models}

A basic hidden Markov model comprises two stochastic processes: an observed process $\{Y_t\}$ and a hidden process $\{S_t\}$. The unobserved process is an $N$-state Markov chain on the state space $\{1, \ldots, N\}$ that satisfies the Markov property $$\Pr({S}_{t} = {s}_{t} \mid {S}_{t-1} = {s}_{t-1},\ldots,{S}_1 = {s}_1) = \Pr({S}_{t} = {s}_{t} \mid {S}_{t-1} = {s}_{t-1}).$$
Due to this Markovian nature, the state process is fully specified by its initial distribution $\bm{\delta}^{(1)} = \bigl(\Pr(S_1 = 1), \dotsc, \Pr(S_1 = N) \bigr)$ and one-step transition probabilities $\gamma_{ij}^{(t)} = \Pr(S_{t} = j \mid S_{t-1} = i)$ which are commonly summarised in a transition probability matrix (t.p.m.) 
$$
\bm{\Gamma}^{(t)} = (\gamma_{ij}^{(t)})_{i,j = 1, \dots, N}.
$$
The time index $t$ is included here because in many applications covariates $\bm{z}_t \in \mathbb{R}^p$, $t = 1, \dotsc, T$, affect the probability of transitioning from one state to another.
If at time $t-1$, the process is in state $i$, the categorical distribution of states for time $t$ is given by $(\gamma_{i1}^{(t)}, \dots, \gamma_{iN}^{(t)})$. The covariate influence on this categorical distribution is usually modelled using a multinomial logistic regression, achieved by appling the inverse multinomial logistic link function (softmax) to linear predictors for each row of the t.p.m.,
$$
\gamma_{ij}^{(t)} = \frac{\exp(\eta_{ij}^{(t)})}{\sum_{k=1}^N \exp(\eta_{ik}^{(t)})},
$$
setting $\eta_{ii}^{(t)} = 0$ (reference category). The linear predictors $\eta_{ij}^{(t)}$ may include any type of parametric relationship (e.g.\ linear, quadratic, interaction), simple random effects or, as will be considered here, \emph{random fields}.

Given a realisation of the hidden state $S_t$, the distribution of $Y_t$ is independent of all other observations and hidden states. More formally, if $f$ denotes a general density,
$$
f(y_t \mid s_1, \dotsc, s_T, y_1, \dotsc, y_{t-1}, y_{t+1}, \dots, y_T) = f(y_t \mid s_t).
$$
As $S_t$ only takes on finitely many values, to simplify notation we simply denote the corresponding densities as $f_j(y_t) = f(y_t \mid S_t = j)$, $j = 1, \dotsc, N$. If $y_t$ is multivariate, a common assumption is that its components are conditionally indepent given the current state. Consequently, in this case $f_j(y_t)$ is the product of univariate densities.

\subsection{Gaussian processes with sparse precision}
\label{subsec:GPs}

Gaussian processes (GPs) provide a flexible framework for function modelling by specifying their distribution through a mean function and a covariance kernel. 
The defining characteristic of a \emph{Gaussian} process $u$ is that for any finite collection of evaluation points $r_1, \dots, r_l$, the corresponding function values follow a multivariate \emph{Gaussian} distribution with mean vector and covariance matrix determined by the mean function $m$ and covariance kernel $\mathcal{K}$, respectively. More formally,
\begin{equation}
\bigl(u(r_1), \dots, u(r_l)\bigr)^\top \sim 
\mathcal{N}(\bm m, \bm{K}),
\end{equation}
where $\bm m = \bigl(m(r_1), \dots, m(r_l)\bigr)^\top$, and $\bm{K}_{ij} = \mathcal{K}(r_i, r_j)$.

In spatial and spatio-temporal settings, this framework is commonly referred to as a Gaussian random field (GRF).
A major practical limitation of GP models is their computational cost. For $l$ evaluation points, likelihood evaluation and inference typically require $\mathcal{O}(l^3)$ operations and $\mathcal{O}(l^2)$ storage due to dense covariance matrices. 


However, \citet{lindgren2011explicit} showed that Gaussian fields with Matérn covariance can be represented as the weak solution to the stochastic partial differential equation (SPDE)
\begin{equation}
(\kappa^2 - \Delta)^{\alpha/2} u(\bm{r})
= \mathcal{W}(\bm{r}),
\label{eq:spde}
\end{equation}
where $\Delta$ denotes the Laplacian operator, $\mathcal{W}(\bm{r})$ is Gaussian white noise, $\kappa>0$ controls the spatial range, and $\bm r \in \mathbb{R}^d$. The parameter $\alpha$ determines the smoothness of the field and is usually fixed as it is poorly identified. For spatial dimension $d$, the corresponding Matérn smoothness parameter is given by $\nu = \alpha - d/2$.

Approximating the weak solution of \eqref{eq:spde} using a finite element basis representation,
$$
u(\bm{r}) \approx \sum_{i=1}^l \psi_i(\bm{r}) \, x_i,
$$
where $\{\psi_i\}_{i = 1, \dots, l}$ are locally supported basis functions defined on a triangulated mesh, yields a finite-dimensional Gaussian approximation of the Matérn field. Due to the local support of the basis functions and the locality of the differential operator, the resulting system matrices are sparse. Consequently, the coefficient vector $\bm{x} = (x_1,\dots, x_l)^\top$ follows a Gaussian Markov random field (GMRF) with sparse precision matrix. 
For example, for dimension $d = 2$ and $\alpha = 2$, we obtain the precision matrix
\begin{equation}
\label{eqn:SPDE_precision}
    \bm{Q}_{\tau, \kappa} = \tau^2 \bigl(\kappa^4 \bm{C} + 2 \kappa^2 \bm{G} + \bm{G}\bm{C}^{-1} \bm{G} \bigr),
\end{equation}
where $\tau > 0$ is an overall precision parameter that scales the driving white noise and thereby determines the marginal variance of the resulting field.
The $l \times l$ matrices $\bm{C}$ and $\bm{G}$ have entries
$$
C_{ij} = \int \psi_i(\bm{r}) \psi_j{(\bm{r})} \; d \bm{r},
\quad \text{and} \quad
G_{ij} = \int \nabla \psi_i(\bm{r})^\top \nabla \psi_j{(\bm{r})} \; d \bm{r},
$$
respectively. 
Using piecewise linear basis functions yields sparse matrices $\bm{C}$ and $\bm{G}$. In practice, $\bm C$ is often replaced by its lumped (diagonal) approximation such that $\bm{C}^{-1}$ is sparse (see \citealp{bolin2009wavelet} for justification).

The stochastic weights $\bm{x}$ then have multivariate Gaussian density
\begin{equation}
    \label{eqn:GMRF_density}
    (2\pi)^{-l / 2} \vert\bm{Q}_{\tau, \kappa}\vert^{1/2} \exp\bigl( -\tfrac{1}{2} (\bm{x} - \bm{\mu})^\top \bm{Q}_{\tau, \kappa} (\bm{x} - \bm{\mu})\bigr),
\end{equation}
where $\bm \mu$ is determined by the mean function $m$.
Evaluating this density is computationally efficient because the quadratic form only involves the non-zero entries of $\bm{Q}_{\tau,\kappa}$, and sparse Cholesky factorisation can be used to compute the log-determinant efficiently.

\subsection{Example model formulations}
\label{subsec: examples models}

Before turning to inference, we shall highlight two examples in which combining hidden Markov models (HMMs) with Gaussian fields or Gaussian processes is natural and beneficial.

\textbf{Example 1} Suppose the state-switching process of interest is not observed directly but is contaminated by an additive noise component. A natural model formulation for observations $y_1, \dots, y_T$ is then
$$Y_t = X_t + u_t, \qquad t = 1, \dots,T,$$
where $\{X_t\}$ follows an HMM and represents the latent process of interest. The noise component $u_t$ may be independent white noise, but it may also exhibit temporal structure, for example $u_t = u(t)$ with $u$ following a Gaussian process. The latter allows for smooth trends or temporally correlated disturbances that cannot be captured by independent errors. We consider such a model in the first case study.

\textbf{Example 2} In spatial settings, suppose we may observe not only $y_1, \dots, y_T$ but also corresponding spatial locations $\bm{r}_1, \dots, \bm{r}_T$, each in $\mathbb{R}^2$. The observations may then follow a standard state-dependent model,
$$Y_t \mid \{S_t = j\} \sim f_j(\cdot),$$
where $\{f_j\}_{j=1,\dots,N}$ are parametric state-dependent densities.
While the observation model remains parametric, the state process itself may depend on spatial location. For example, transition probabilities $\gamma_{ij}^{(t)}$ can be linked to linear predictors of the form
$$\eta_{ij}^{(t)} = \bm{\beta}_{ij}^\top \bm{z}_t + u(\bm{r}_t),$$
where $\bm{z}_t$ denotes observed covariates and $u(\bm{r}_t)$ is a spatial Gaussian field evaluated at location $\bm{r}_t$. This formulation allows behavioural switching dynamics to vary smoothly across space. In the second case study, we apply such a model to animal movement data to investigate how behavioural states of lions vary across the study area.

\subsection{Inference via marginal likelihood using (R)TMB}
\label{subsec:inference}

Conducting frequentist inference, i.e.\ maximum likelihood estimation, in any of the aformentioned models requires access to the likelihood function of the observations under the particular model of interest. 
The main challenge is that, although it is usually straightforward to compute the joint likelihood of the data and the latent variables, the latter must be integrated out to obtain the marginal likelihood of the data.
Precisely, letting $f$ denote a general density and using subscripts to denote parameter dependence, we require the marginal likelihood function 
\begin{equation}
\label{eqn: marginal likelihood}
\mathcal{L}(\bm{\theta}) = \int_{\mathbb{R}^l} f_\theta(\bm{y} \mid \bm{x}) f_\theta(\bm{x}) \; d\bm{x}
\end{equation}
of observations $\bm{y} = (y_1, \dots, y_T)^\top$. The vector $\bm{\theta} \in \mathbb{R}^p$ parameterises the t.p.m., the state-dependent distributions and the distribution of the Gaussian field. 
Equation \eqref{eqn: marginal likelihood} covers two distinct cases, corresponding to the example model formulations from Section~\ref{subsec: examples models} and case studies from Section~\ref{sec:case_studies}.
\begin{enumerate}
    \item In Example~1, $f_\theta(\bm{x})$ gives the HMM likelihood of a latent process $\bm{x} = (x_1, \dots, x_T)^\top$ that is perturbed by noise, thus in this case $l = T$. 
    $f_\theta(\bm{y} \mid \bm{x})$ is a GMRF density as in \eqref{eqn:GMRF_density} evaluated at $\bm y$ with mean $\bm x$, corresponding to a measurement model of the observations conditional on the process.
    \item In Example~2, $f_\theta(\bm{y} \mid \bm{x})$ is the HMM likelihood of obervations $y_1, \dots, y_T$ conditional on a latent Gaussian field $\bm{x} \in \mathbb{R}^l$, the latter with GMRF density $f_\theta(\bm{x})$ as in \eqref{eqn:GMRF_density}.
\end{enumerate}

Due to the structure of the HMM likelihood, the high-dimensional integral in Equation \eqref{eqn: marginal likelihood} is analytically intractable. 
A powerful tool 
for approximate frequentist inference is given by the so-called \textit{Laplace approximation} \citep{erkanli1994laplace, van2000asymptotic}. 
Omitting the dependence on the data $\bm{y}$ for notational simplicity and letting
\begin{equation}
\label{eqn:g}
    g(\bm{x}, \bm{\theta}) = - \log f_\theta(\bm{y} \mid \bm{x}) - \log f_\theta(\bm{x}),
\end{equation}
we can rewrite the marginal likelihood in Equation \eqref{eqn: marginal likelihood} as 
\begin{equation}
\label{eqn: marginal likelihood 2}
   \mathcal{L}(\bm{\theta}) = \int_{\mathbb{R}^l} \exp \bigl( - g(\bm{x, \bm{\theta}}) \bigr) \; d\bm{x}. 
\end{equation}

For fixed $\bm{\theta}$ the Laplace approximation now proceeds by performing a second-order Taylor approximation of $g$ about its minimum $\hat{\bm{x}}_{\theta} = \underset{\bm x}{\text{argmin}}\; g(\bm{x}, \bm{\theta})$, yielding the quadratic approximation
$$
g^*(\bm{x}, \bm{\theta}) = g(\hat{\bm{x}}_{\theta}, \bm{\theta}) + \tfrac{1}{2} (\bm{x} - \hat{\bm{x}}_{\theta})^\top \bm{H}_\theta (\bm{x} - \hat{\bm{x}}_{\theta}),
$$
where $\bm{H}_\theta = g^{\prime \prime}_{xx}(\hat{\bm{x}}_{\theta}, \bm{\theta})$ is the Hessian matrix of $g$ w.r.t.\ $\bm{x}$.
Plugging $g^*$ back into \eqref{eqn: marginal likelihood 2}, yields
$$
\mathcal{L}^*(\bm{\theta}) = \exp\bigl(-g(\hat{\bm{x}}_\theta, \bm{\theta} )\bigr) \int_{\mathbb{R}^l} \exp \bigl( -\tfrac{1}{2} (\bm{x} - \hat{\bm{x}}_{\theta})^\top \bm{H}_\theta (\bm{x} - \hat{\bm{x}}_{\theta}) \bigr) \; d\bm{x},
$$
where the latter part is a multivariate Gaussian integral with closed-form solution. Solving the integral, applying the logarithm and multiplying by $-1$, we arrive at the approximate marginal negative log-likelihood function
\begin{equation}
    \label{eqn: laplace approx log-likelihood}
    -\log \mathcal{L}^*(\bm{\theta}) = g(\hat{\bm{x}}_{\theta}, \bm{\theta}) - \tfrac{l}{2} \log(2 \pi) + \tfrac{1}{2} \log \vert \bm{H}_{\theta} \vert,
\end{equation}
where $\vert \cdot \vert$ denotes the determinant.
This procedure is conveniently automated through the Template Model Builder framework and its associated \texttt{R} package \texttt{TMB} \citep{kristensen2016tmb} and has recently become even more attractive through the \texttt{R} interface provided by the novel \texttt{RTMB} package \citep{kristensen2024rtmb}.
With the latter, the user simply needs to implement and supply the joint negative log-likelihood $g$ as an \texttt{R} function. 
Using reverse-mode automatic differentiation for the derivatives of $g$, the package then returns \eqref{eqn: laplace approx log-likelihood} along with its gradient w.r.t.\ $\bm\theta$. Each call to the objective function (or its gradient) then triggers an inner Newton optimisation to determine $\hat{\bm{x}}_{\theta}$, required for the Laplace approximation. Fitting the model then amounts to numerically minimising \eqref{eqn: laplace approx log-likelihood} using any suitable gradient based numerical minimisation routine, yielding a nested numerical optimisation procedure.
Due to the nature of the Laplace approximation, after the optimisation predicted latent variables $\hat{\bm x}_{\hat{\theta}}$, corresponding to the posterior mode, are immediately available. 
Furthermore, Template Model Builder provides an estimate of the joint precision matrix of the parameters and latent variables \citep{bates1988nonlinear}. Using this matrix, we can sample from the approximate Gaussian distribution and use the aquired samples to quantify uncertainty.

A crucial practical consideration is the structure of $\bm{H}_{\theta}$.
If the dimensionality $l$ of $\bm{x}$ is large, computation of $\bm{H}_{\theta}$ becomes prohibitive in general.
Only if many entries of $\bm{H}_{\theta}$ are exactly zero, i.e.\ it is a \emph{sparse} matrix, computations become feasible and fast. An entry $H_{ij}$ is zero if $\partial^2 g / \partial x_i \partial x_j = 0$ everywhere.
This requirement makes the SPDE approach outlined in Section \ref{subsec:GPs} particularly attractive. The term corresponding to the GMRF log-likelihood will be the logarithm of $\eqref{eqn:GMRF_density}$ and hence take the form
$$
-\tfrac{1}{2} (\bm{x} - \bm{\mu})^\top \bm{Q}_{\tau, \kappa} (\bm{x} - \bm{\mu}) + \text{const.}
$$
The second derivative of this summand w.r.t.\ $\bm{x}$ (or $\bm{\mu}$) is simply $- \bm{Q}_{\tau, \kappa}$, a sparse matrix by construction.
If such a pattern is present, it is automatically detected by (\texttt{R})\texttt{TMB} and exploited for all required computations. Specifically, only non-zero cross-derivatives are computed and the linear algebra required to compute \eqref{eqn: laplace approx log-likelihood} is performed using a sparse Cholesky factor.

In either of the two cases from Section~\ref{subsec: examples models}, the second summand in \eqref{eqn:g} corresponds to the HMM log-likelihood. Evaluating this term requires marginalisation over the latent state sequence $s_1, \dots, s_T$ which, as we shall see, results in a dense Hessian matrix.
Naively trying to apply the Laplace approximation is hence deemed to fail as evaluating $\bm{H}_\theta$ and performing computations with it is prohibitive.
Hence, in the next section we develop an approximation to the algorithm required to evaluate the HMM log-likelihood in order to mitigate a dense Hessian matrix. 

\subsection{A banded forward algorithm}
\label{subsec: forward algo}

Given our two examples, we either need to compute the HMM likelihood of a latent process $x_1, \dots, x_T$, or of observations $y_1, \dots, y_T$ conditional on a random field $\bm{x} \in \mathbb{R}^l$.
To simplify notation, in this section we focus on the first case and use the term ``observations'' for $x_1, \dots, x_T$.
Note, however, that this does not restrict generality. The main objective is to construct a log-likelihood function that yields zero cross-derivatives w.r.t.\ any two quantities entering log-likelihood summands at time $t$ and $t'$ if $\vert t - t'\vert$ is sufficiently large.

Computing the likelihood of $x_1, \dotsc, x_T$ under an HMM requires marginalising over all possible unobserved states $s_1, \dots, s_T$ that could have generated this state-dependent sequence. 
The task of summing out the states is made feasible through the so-called forward algorithm \citep{zucchini2016hidden}. Defining the \emph{forward variables}
$$
\alpha_t(j) = f(x_1, \dots, x_t, s_t = j), \quad \bm\alpha_t = \bigl(\alpha_t(1), \dots, \alpha_t(N) \bigr),
$$
and the diagonal matrix $\bm{P}(x_t) = \text{diag}\bigl( f_1(x_t), \dots, f_N(x_t) \bigr)$, we have
$$
\bm\alpha_1 = \bm\delta^{(1)} \bm{P}(x_1), \quad \text{and} \quad \bm\alpha_t = \bm\alpha_{t-1} \bm{\Gamma}^{(t)}\bm{P}(x_t).
$$
Running the recursion from $t = 1$ to $t = T$ and building the sum of the entries of the last forward variable yields the likelihood of interest. However, this approach can suffer from numerical under- or overflow, and thus it is common to consider the following scaling strategy \citep{lystig2002exact}. Let $\bm{1} = (1, \dots, 1)$ and define the \emph{scaled} forward variables
$$
\bm\phi_t = \frac{\bm\alpha_t}{\bm\alpha_t \bm{1}^\top},
$$
and thus $\phi_t(j) = f(s_t=j \mid x_1, \dots, x_t)$ by the definition of $\alpha_t(j)$. Now it can be shown that
$$
\bm\phi_t = \frac{\bm\phi_{t-1} \bm{\Gamma}^{(t)}\bm{P}(x_t)}{\bm\phi_{t-1} \bm{\Gamma}^{(t)}\bm{P}(x_t) \bm{1}^\top}, \quad \text{and} \quad
\bm\phi_{t-1} \bm{\Gamma}^{(t)}\bm{P}(x_t) \bm{1}^\top = f(x_t \mid x_1, \dots, x_{t-1}).
$$
Observing that $\bm\phi_t$, $t = 1, \dots, T$, can also be calculated recursively, this yields an alternative way to calculate the likelihood based on the product rule
$$
f(x_1, \dots, x_T) = \underbrace{f(x_1)}_{\bm{\delta}^{(1)} \bm{P}(x_1) \bm{1}^\top} \prod_{t=2}^T \underbrace{f(x_t \mid x_1, \dots, x_{t-1})}_{\bm\phi_{t-1} \bm{\Gamma}^{(t)}\bm{P}(x_t) \bm{1}^\top},
$$
and taking logarithms, we obtain a summation over all time points:
\begin{equation}
\label{eqn: log_likelihood}
    \log f(x_1, \dots, x_T) = \log \underbrace{f(x_1)}_{\bm{\delta}^{(1)} \bm{P}(x_1) \bm{1}^\top} + \sum_{t=2}^T \log \underbrace{f(x_t \mid x_1, \dots, x_{t-1})}_{\bm\phi_{t-1} \bm{\Gamma}^{(t)}\bm{P}(x_t) \bm{1}^\top}.
\end{equation}
The algorithm for computing \eqref{eqn: log_likelihood} then involves a recursion over time where at each time point we a) update the scaled forward variable based on the previous one and b) add the corresponding log-likelihood contribution to a running summation.
The key observation here is that each term in the summation in \eqref{eqn: log_likelihood} is conditional on the entire process history through the scaled forward variable $\bm{\phi}_{t-1}$. 
Consequently, 
the state-dependent process of a hidden Markov model is not Markovian.
Hence, the Hessian of \eqref{eqn: log_likelihood} with respect to $\bm{x}$ is a dense matrix.

This second point extends more generally to derivatives taken with respect to parameters that enter only one or a small number of summands in \eqref{eqn: log_likelihood}. 
Such a situation arises in the second example, where we consider the HMM likelihood of an observed sequence $y_1, \dots, y_T$, conditional on the GMRF $\bm x = (x_1, \dotsc, x_l)^\top$ based on the SPDE approach.
Each weight is associated with a basis function of compact support. 
Indeed, at any spatial location $\bm r$, at most three basis functions $\psi_i(\bm r)$ are non-zero, giving an index set $I(\bm r)$ with at most three entries. The summand in the log-likelihood for time $t$ can then be written as
$$
\log \bigl( \bm{\phi}_{t-1} \bm \Gamma(x_{i_1}, x_{i_2}, x_{i_3}) \bm P (y_t) \bm 1^\top \bigr), \quad i_1, i_2, i_3 \in I(\bm r_t),
$$
only depending on a small subset of the GMRF weights. However, the forward variable $\bm\phi_{t-1}$ carries information from all previous time points. Without further modification, derivatives with respect to field weights appearing at different times can interact through the recursion, producing a dense Hessian unless modifications are applied.



To achieve independence of terms sufficiently far apart in time, we split the sum in \eqref{eqn: log_likelihood} into $B$ non-overlapping blocks of equal length $k$. We refer to $k$ as the \emph{bandwidth} parameter.
To simplify notation, we assume that $T = Bk$ and define
$t_b = (b-1)k$. The log-likelihood can then be written as
$$
\log f(x_1,\dots,x_T)
=
\log f(x_1,\dots,x_k)
+
\sum_{b=2}^B
\sum_{t=t_b+1}^{bk}
\log \underbrace{f(x_t \mid x_1, \dots, x_{t-1})}_{\bm\phi_{t-1} \bm{\Gamma}^{(t)}\bm{P}(x_t) \bm{1}^\top}.
$$

\begin{figure}
    \centering
    \begin{tikzpicture}

    \coordinate (A) at (0,0);
    \coordinate (B) at (2,-1.6);
    \coordinate (C) at (4,-3.2);
    \coordinate (D) at (6,-4.8);

     \draw [dashed] (A) -- ($(A) + (2,0)$);
     \draw ($(A) + (2,0)$) -- ($(A) + (4,0)$);
     
     \draw ($(A) + (0,-0.1)$) -- ($(A) + (0,0.1)$);
     \draw ($(A) + (2,-0.1)$) -- ($(A) + (2,0.1)$);
     \draw ($(A) + (4,-0.1)$) -- ($(A) + (4,0.1)$);

     \draw[decorate, decoration={brace, amplitude=4pt}]
     ($(A) + (2,0.3)$) -- ($(A) + (4,0.3)$) 
     node[midway, above=5pt] {\small $f(x_{k+1}, \dots, x_{2k} \mid x_1, \dots, x_k)$};

     \draw[-{Latex}] ($(A) + (0.4,-0.4)$) -- ($(A) + (1.6,-0.4)$);
     \node at ($(A) + (0,-0.4)$) {$\bm{\delta}^{(1)}$};
     \node at ($(A) + (2,-0.4)$) {$\bm{\phi}_k$};

     \draw [dashed] (B) -- ($(B) + (2,0)$);
     \draw ($(B) + (2,0)$) -- ($(B) + (4,0)$);

     \draw ($(B) + (0,-0.1)$) -- ($(B) + (0,0.1)$);
     \draw ($(B) + (2,-0.1)$) -- ($(B) + (2,0.1)$);
     \draw ($(B) + (4,-0.1)$) -- ($(B) + (4,0.1)$);

     \draw[decorate, decoration={brace, amplitude=4pt}]
     ($(B) + (2,0.3)$) -- ($(B) + (4,0.3)$) 
     node[midway, above=5pt] {\small $f(x_{2k+1}, \dots, x_{3k} \mid x_{k+1}, \dots, x_{2k})$};

     \draw[-{Latex}] ($(B) + (0.4,-0.4)$) -- ($(B) + (1.6,-0.4)$);
     \node at ($(B) + (0,-0.4)$) {$\bm{\rho}$};
     \node at ($(B) + (2,-0.4)$) {$\tilde{\bm{\phi}}^k_{2k}$};

     \draw [dashed] (C) -- ($(C) + (2,0)$);
     \draw ($(C) + (2,0)$) -- ($(C) + (4,0)$);
     
     \draw ($(C) + (0,-0.1)$) -- ($(C) + (0,0.1)$);
     \draw ($(C) + (2,-0.1)$) -- ($(C) + (2,0.1)$);
     \draw ($(C) + (4,-0.1)$) -- ($(C) + (4,0.1)$);

     \draw[decorate, decoration={brace, amplitude=4pt}]
     ($(C) + (2,0.3)$) -- ($(C) + (4,0.3)$) 
     node[midway, above=5pt] {\small $f(x_{3k+1}, \dots, x_{4k} \mid x_{2k+1}, \dots, x_{3k})$};

     \draw[-{Latex}] ($(C) + (0.4,-0.4)$) -- ($(C) + (1.6,-0.4)$);
     \node at ($(C) + (0,-0.4)$) {$\bm{\rho}$};
     \node at ($(C) + (2,-0.4)$) {$\tilde{\bm{\phi}}^k_{3k}$};

     \node at ($(D) + (2,0.5)$) {$\dots$};
    
    \end{tikzpicture}

    \caption{Visual demonstration of the approximation employed for computing the banded forward algorithm. Each block is used twice: 1) for computing its likelihood contribution and 2) for constructing an approximate scaled forward variable for the next block's likelihood contribution.}
    \label{fig:forward}
\end{figure}
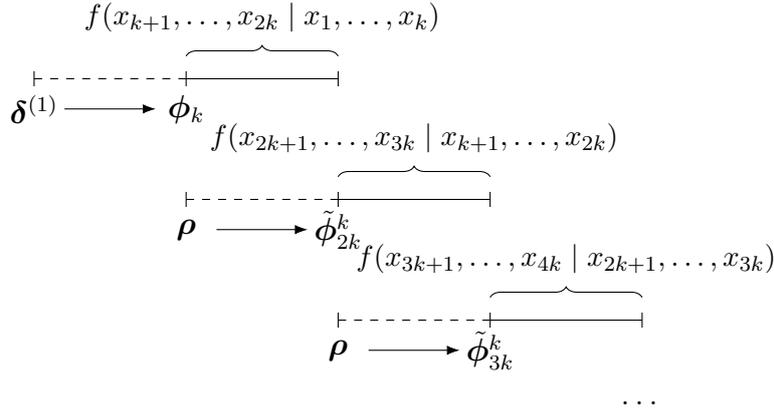

The first block, $\log f(x_1,\dots,x_k)$, can be evaluated exactly, with the forward recursion initialised using the initial distribution $\bm{\delta}^{(1)}$, yielding the scaled forward variable $\bm{\phi}_k$.
This quantity can then also be used to evaluate the log-likelihood contribution of the second block exactly.

The key difference arises when evaluating the (approximate) log-likelihood contributions for blocks $b \geq 3$.
Usually, each corresponding recursion needs to be initialised with the scaled forward variable $\bm{\phi}_{t_b}$ which depends on $x_1, \dots, x_{t_b}$. 
However, to achieve independence of $x_1, \dots, x_{t_{b-1}}$, we construct an approximation $\tilde{\bm{\phi}}^k_{t_b}$. To do so, we define a fixed probability vector $\bm{\rho}$, e.g.\ $\bm{\rho} = \left(\tfrac{1}{N},\dots,\tfrac{1}{N}\right)$, and propagate it through the forward recursion over the \emph{previous} block only. That is, at time $t_b$, we replace the exact scaled forward variable at the start of the block, $\bm{\phi}_{t_b}$, by 
$$
\tilde{\bm{\phi}}^k_{t_b}
=
\frac{
\bm\rho
\prod_{t=t_{b-1}+1}^{t_b}
\bm{\Gamma}^{(t)} \bm{P}(x_t)
}{
\bm\rho
\prod_{t=t_{b-1}+1}^{t_b}
\bm{\Gamma}^{(t)} \bm{P}(x_t)
\bm{1}^\top
},
$$
with entries $\tilde{\phi}_{t_b}(j) = f(s_{t_b} = j \mid x_{t_{b-1}+1}, \dotsc, x_{t_b})$.
In doing so, for each block we \emph{truncate} the conditioning on the process history 
and hence obtain the approximate log-likelihood function
\begin{equation}
\label{eqn:log_likelihood_approx}
\log f(x_1,\dots,x_T)
\approx
\log f(x_1,\dots,x_k)
+
\sum_{b=2}^B
\sum_{t=t_b+1}^{bk}
\log \underbrace{f(x_t \mid x_{t_{b-1}+1}, \dots, x_{t-1})}_{\tilde{\bm\phi}_{t-1} \bm{\Gamma}^{(t)}\bm{P}(x_t) \bm{1}^\top}.
\end{equation}

\begin{algorithm}[t]
\caption{Forward Recursion for log-likelihood and scaled forward variable}
\label{alg:forward_recursion}
\begin{algorithmic}[1]
\Require Observations $x_1, \dots, x_T$, transition matrices $\bm\Gamma^{(t)}$, state-dependent densities $f_j(x_t)$, initial state vector $\bm\rho$
\Ensure Log-likelihood $l$ and last scaled forward variable $\bm{\phi}$

\State Initialise scaled forward variable:
$$
\bm{\phi} \gets \frac{\bm\rho \bm{P}(x_1)}{\bm\rho \bm{P}(x_1) \bm{1}^\top}
$$
\State Initialise log-likelihood:
$$
l \gets \log(\bm\rho \bm{P}(x_1) \bm{1}^\top)
$$

\For{$t = 2$ \textbf{to} $T$}
    \State Update scaled forward variable:
    $$
    \bm{\phi} \gets \frac{\bm{\phi} \bm\Gamma^{(t)} \bm{P}(x_t)}
                            {\bm{\phi} \bm\Gamma^{(t)} \bm{P}(x_t) \bm{1}^\top}
    $$
    \State Update log-likelihood:
    $$
    l \gets l + \log (\bm{\phi} \bm\Gamma^{(t)} \bm{P}(x_t) \bm{1}^\top)
    $$
\EndFor

\State \Return $(l, \bm{\phi})$
\end{algorithmic}
\end{algorithm}

\begin{algorithm}[t]
\caption{Banded Forward Algorithm for Approximate log-likelihood}
\label{alg:banded_forward}
\begin{algorithmic}[1]
\Require Observations $x_1, \dots, x_T$, transition matrices $\bm\Gamma^{(t)}$, state-dependent densities $f_j(x_t)$, initial distribution $\bm\delta^{(1)}$, bandwidth $k$, fixed initialisation vector $\bm{\rho}$
\Ensure Approximate log-likelihood $l$

\State $B \gets T / k$  \Comment{Assume $T$ is a multiple of $k$ for simplicity}
\State Compute exact log-likelihood and final scaled forward variable of first block:
\[
(l, \bm{\phi}) \gets \text{ForwardRecursion}(x_1, \dots, x_k, \bm\Gamma^{(t)}, f_j, \bm\delta^{(1)})
\]
\State Compute exact log-likelihood for second block
$$
l \gets l + \text{ForwardRecursion}(x_{k+1}, \dots, x_{2k}, \bm\Gamma^{(t)}, f_j, \bm{\phi})
$$

\For{$b = 3$ \textbf{to} $B$}
    \State Compute approximate scaled forward variable from previous block:
    $$
    \tilde{\bm{\phi}} \gets \text{ForwardRecursion}(x_{(b-2)k+1}, \dots, x_{(b-1)k}, \bm\Gamma^{(t)}, f_j, \bm{\rho})
    $$
    \State Update log-likelihood for current block:
    $$
    l \gets l + \text{ForwardRecursion}(x_{(b-1)k+1}, \dots, x_{bk}, \bm\Gamma^{(t)}, f_j, \tilde{\bm{\phi}})
    $$
\EndFor

\State \Return $l$
\end{algorithmic}
\end{algorithm}

The precise algorithms for the forward recursion and the banded forward algorithm are given in Algorithms \ref{alg:forward_recursion} and \ref{alg:banded_forward}, respectively. 
Visual intution is given by Figure~\ref{fig:forward}. 
As each block is used \emph{twice} to run the forward recursion, the computational cost of the algorithm is about twice that of the cost of the standard forward algorithm. In high-dimensional random-effect settings, this is however greatly compensated by the fact that it generates a banded Hessian matrix as shown in Figure~\ref{fig:hessians}. 

Motivation for this local approximation is given by the fact that Markov chains forget exponentially fast about their initial condition.
\begin{figure}
    \centering
    \includegraphics[width=1\linewidth]{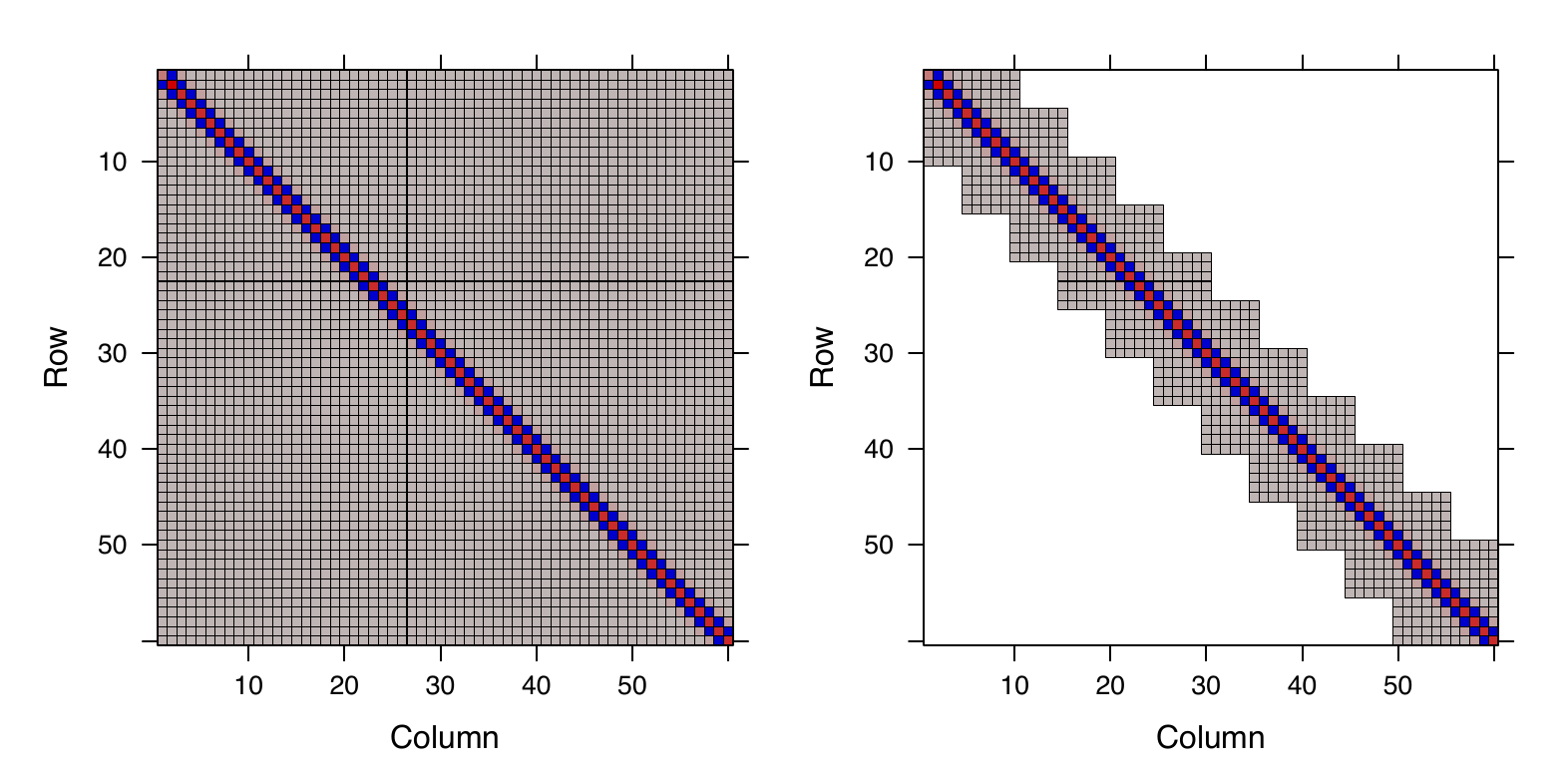}
    \caption{Comparison of the Hessian matrix w.r.t.\ observation sequence $x_1, \dots, x_{60}$ when the log-likelihood is calculated using the regular forward algorithm (left) and using the banded forward algorithm with bandwidth $k = 5$ (right). Gray corresponds to values that are very close to zero in magnitude while white corresponds to exact zeros.}
    \label{fig:hessians}
\end{figure}
Hence, 
for a sufficiently large lag $k$, we expect $\tilde{\bm{\phi}}^k_t$ to be indistinguishable from $\bm{\phi}_t$ regarding numerical precision, and thus will be the conditional likelihood contribution. 
Formal theoretical justification for the algorithm is given by the following theorem.
\begin{theorem}[Geometric decay of the forward-likelihood approximation error]
\label{thm1}
Let $\ell_T(\bm{\theta}, \bm{x})$ denote the exact log-likelihood and 
$\tilde{\ell}_T(\bm{\theta}, \bm{x})$ its bandwidth-$k$ forward approximation. 
Consider either of the following settings:
\begin{enumerate}
    \item[(i)] \textbf{Latent-process case:} 
    $\ell_T(\bm{\theta}, \bm{x})$ is the log-likelihood of a latent process 
    $\bm{x} = (x_1,\dots,x_T)^\top$.
    \item[(ii)] \textbf{Observed-data case:} 
    $\ell_T(\bm{\theta}, \bm{x})$ is the log-likelihood of observations 
    $y_1,\dots,y_T$ conditional on latent variables $\bm{x}$.
\end{enumerate}

Assume that $\bm\theta \in \bm\Theta$, $\bm x \in \mathcal{X}$, $\bm{\Theta} \times \mathcal{X}$ is compact, the state-dependent densities satisfy $0 < m \le f_j(\cdot) \le M < \infty$ for all $j,t$, 
and the underlying Markov chain is uniformly ergodic.

Then there exist constants $C>0$ and $\rho \in (0,1)$, independent of $k$, such that
$$
\sup_{(\bm{\theta}, \bm{x}) \in \bm{\Theta} \times \mathcal{X}} 
\big| \ell_T(\bm{\theta},\bm{x}) - \tilde{\ell}_T(\bm{\theta},\bm{x}) \big|
\le T C \rho^k,
$$
where the inequality holds almost surely in case (ii).
\end{theorem}
For the proof, see Appendix~\ref{a1:proof}.
While Theorem~\ref{thm1} yields a theoretical guarantee that the approximation error can be made arbitrarily small by choosing a sufficiently large $k$, in practice, the constants $C$ and $\rho$ can vary substantially with the characteristics of the data at hand. For good practical performance --- balancing approximation accuracy and numerical efficiency --- a suitable bandwidth should therefore always be tuned to the specific application.
We study the practical implications of varying $k$, and overall estimation performance in the simulation experiments in Appendix~\ref{a3:sim}. Overall these indicate adequate recovery of the true spatial field and all other model parameters. In these experiments, the impact of the bandwidth parameter was relatively modest. However, covering all possible real-data settings is impossible, hence it is difficult to give general guidance.

In practice, we expect values around $k = 15$ to be useful initial guesses. For highly persistent processes, i.e., when diagonal transition probability matrix entries are large, a larger bandwidth is likely required.
A useful strategy for choosing $k$ is to compute the joint negative log-likelihood $g$ at the initial parameter values and gradually increase the bandwidth $k$ until this value stabilises. This provides a sensible starting point for the first model fit. After this initial fit, a sanity check is recommended: analysts should increase the bandwidth by, e.g., $5$, and verify that the estimated parameters remain equal (to a reasonable number of decimal places). If they do, this indicates that the approximation is sufficiently accurate.






\subsection{Practial implementation}

We use the automatic Laplace approximation implemented in the \texttt{RTMB} package \citep{kristensen2024rtmb}, choosing \texttt{RTMB} over \texttt{TMB} simply because defining custom likelihood functions in \texttt{R} is usually more convenient than writing the required \texttt{C++} code necessary for \texttt{TMB}. 

The approximate forward algorithm from Section \ref{subsec: forward algo} is implemented in the \texttt{LaMa} \texttt{R} package \citep{koslikLaMa2025}, a modular framework for building and fitting latent Markovian models \citep{mews2025build}.
For the two core functions that run the homogeneous and inhomogeneous forward algorithm, \texttt{forward()} and \texttt{forward\_g()}, respectively, the user simply needs to pass a \texttt{bw} argument, corresponding to the bandwidth parameter $k$.

For the SPDE approach, we rely on the \texttt{fmesher} \texttt{R} package \citep{lindgren2025fmesher} to construct the required triangulated mesh and finite element matrices. Efficient computation of the GP likelihood is possible through \texttt{RTMB}'s \texttt{dgmrf()} function which evaluates a multivariate Gaussian density parameterised in terms of a sparse precision matrix, and thus works seamlessly with the sparse matrices provided by \texttt{fmesher}.

\section{Case studies}

\label{sec:case_studies}

The following two case studies demonstrate the practical usage of the proposed method for the model formulations presented in Section \ref{subsec: examples models}.
All code and data for full reproducibility of the analyses is available at \url{https://github.com/janolefi/HMMs_GFs}.

\subsection{Detecting stellar flares in photometric data}

We revisit the analysis of \citet{esquivel2025detecting}, focussing on stellar flare detection in the brightness measurements of M dwarf TIC 031381302 observed by TESS at a 2-minute cadence.
Our goal is to demonstrate the computational advantages of our frequentist estimation approach over the dynamic Hamiltonian Monte Carlo algorithm used in their analysis.
Stellar flares are sudden bursts of energy emitted by stars, primarily caused by magnetic reconnection events. Such events are crucial for understanding stellar magnetic activity, rotation, and the radiation environment of orbiting exoplanets. Hence, detecting and characterising flares in photometric time series data is essential for studying stellar behavior.

Flares are characterised by a rapid increase in brightness followed by a gradual decay. Detection is however complicated by the fact that this is usually superimposed on quasi-periodic oscillations in the star's brightness. Traditional methods for flare detection, such as sigma-clipping, often struggle to distinguish low-energy flares from these oscillations, leading to biased energy estimates. 
To address these challenges, \citet{esquivel2025detecting} introduced a $3$-state HMM, modelling the star’s light curve as transitioning between ``quiet'', ``firing'', and ``decaying'' states, superimposed with an additive quasi-periodic trend.
We retain the core structure of their model, where the observed light curve $\{Y_t\}$ is decomposed into a flaring channel $\{X_t\}$ and a quasi-periodic trend $\{u_t\}$ as
$$
Y_t = X_t + u_t, \quad t = 1, \dots, T.
$$
We model $u_t$ as a Gaussian process with scalar mean $\mu$ and, slightly deviating from \citet{esquivel2025detecting}, with a Matérn GP via the SPDE approach. In order to mirror the rotation kernel employed by \citet{esquivel2025detecting}, 
we use the extension to \emph{oscillating} covariance functions, introduced in Section 3.3 of \citet{lindgren2011explicit}. 
This specification results in a precision matrix of the form
$$
\bm{Q}_{\tau, \kappa, \omega} = \tau^2 \bigl(\kappa^4 \bm{C} + 2 \cos(\pi \omega) \kappa^2 \bm{G} + \bm{G}\bm{C}^{-1} \bm{G} \bigr)
$$
with matrices $\bm{C}$ and $\bm{G}$ as in Section~\ref{subsec:GPs} and oscillation parameter $\omega \in (0,1)$.
The corresponding oscillating covariance function is given in Appendix~\ref{a4:oscGP}.
For this GP with 1D domain, mesh construction is simple. Using the \texttt{fmesher} package, we construct degree one B-spline basis functions (tent functions) with knots at the observed time points and compute the finite element matrices.

Again following \citet{esquivel2025detecting}, for the states ``quiet'' (Q), ``firing'' (F), and ``decaying'' (D), we assume the (partially autogressive) state-dependent model
\begin{align*}
X_t &\mid \{S_t = Q\} \sim \mathcal{N}(0, \sigma^2),\\
X_t &\mid \{S_t = F\}, X_{t-1} \sim \mathcal{N}(X_{t-1}, \sigma^2) + \text{Exp}(\lambda),\\
X_t &\mid \{S_t = D\}, X_{t-1} \sim \mathcal{N}(r X_{t-1}, \sigma^2)
\end{align*}
The state process model is a homogeneous Markov chain, where transitions between quiet (Q) and decaying (D) and transitions between firing (F) and quiet (Q) are forbidden. 
Transitions from decaying (D) to firing (F) are permitted, with such a behaviour being called a \emph{compound flare}. Hence, the model has transition matrix
$$
\bm{\Gamma} = \begin{pmatrix}
    \gamma_{QQ} & \gamma_{QF} & 0\\
    0 & \gamma_{FF} & \gamma_{FD} \\
    \gamma_{DQ} & \gamma_{DF} & \gamma_{DD}
\end{pmatrix},
$$
with all non-zero entries being estimated freely. Note that, while $\bm\Gamma$ has zero entries, Theorem~\ref{thm1} still applies because all entries of $\bm\Gamma^2$ are strictly positive, i.e.\ the Markov chain is ergodic with index of primitivity 2.

We implemented a custom likelihood function corresponding to the model above via the 
\texttt{LaMa} \texttt{R} package. 
We split the time series in two because of a long sequence of missing values in the middle.
Finally, the unobserved true process is integrated out using the automatic Laplace approximation of the \texttt{RTMB} package. 
The model is fitted by numerically optimising the Laplace-approximated marginal (negative log-) likelihood function using \texttt{R}'s \texttt{nlminb()}.
Fitting the model to the full remaining data set of two time series and a total of $T = 19,259$ observations with a bandwidth parameter of  $k = 20$ took about 6 minutes on a single core of an Apple M2 chip with 16GB of memory. 
This represents a dramatic improvement in computational efficiency compared to the 1--4 hours per chunk reported by \citet{esquivel2025detecting}. Their approach required splitting the time series into chunks of 2,000 observations due to the computational complexity of their Gaussian process (GP) implementation. In contrast, our method leverages the linear scalability of the SPDE-based approach, enabling efficient inference on the entire time series without the need for chunking.

\begin{figure}
    \centering
    \includegraphics[width=1\linewidth]{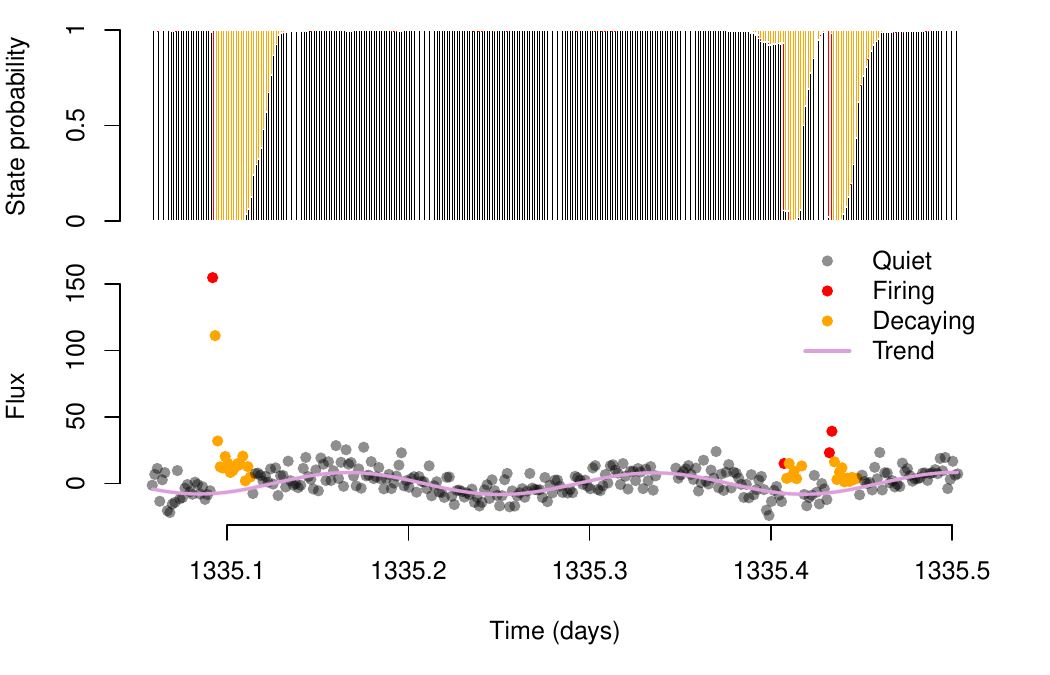}
    \caption{Top panel: Stacked locally decoded state probabilities ($\Pr(S_t = i \mid \bm y)$) for states quiet (black), firing (red) and decaying (orange) for a subsection of the brightness time series. Bottom panel: Corresponding brightness values, colour-coded according to the Viterbi-decoded (most-likely) state \citep{viterbi2003error}. The purple line is the estimated quasi-periodic trend, i.e.\ posterior mode of the Gaussian process. 
    We also show an approximate 95\% credible interval for the fitted smooth in light purple, which is barely visible due to the high estimation precision.}
    \label{fig:flares}
\end{figure}

For the time series of M dwarf TIC 031381302, the model flagged 18 flare events, with an average duration of 17.4 minutes including the decaying phase. We find that the model is able to detect flare events while simultaneously controlling for the quasi-periodic trend in oscillation (cf.\ Figure~\ref{fig:flares}).
Specifically, note that the second flare in Figure~\ref{fig:flares} would be very difficult to detect using a simpler model without a quasi-periodic trend, as the brightness value of the observations decoded as ``flaring'' is lower than that of previous observations that are decoded as ``quiet''.
The full time series with detected flares is shown in Appendix~\ref{a5:figs}.

We do not aim to provide a formal performance comparison with the approach of \citet{esquivel2025detecting}, nor to claim that our specification achieves the same level of physical validity. Their work carefully investigated model adequacy and detection performance using injection experiments, which is beyond the scope of this case study. Instead, our objective is to illustrate that a closely related modelling framework can be fitted using a Laplace-approximated marginal likelihood combined with an SPDE-based GP representation, resulting in substantially reduced computational load. The example highlights the potential of this frequentist inference strategy as a scalable and practically attractive alternative. 


\subsection{Lion movement in Central Kalahari}



In our second case study, we consider GPS movement data of lions (\textit{Panthera leo}) in the Central Kalahari, Botswana, publicly available through the Movebank Data Repository \citep{MacFarlane2014}. 
Data were collected from December 2008 to March 2012 and originally comprised 241,858 observations. 
Tracks from sensors with long, systematically missing periods not missing at random (e.g., no data between 8~a.m. and 3~p.m.) were excluded. 
The remaining near-hourly data were regularised to exact hourly positions by aligning observations within 15~minutes of the full hour to the hour and discarding all others. This resulted in a final dataset of $T = 53,453$ hourly positions. 
The GPS measurements were converted into \emph{step lengths} (km) and \emph{turning angles} (radians) as it is commonly done when analysing horizontal movement data \citep{langrock2012flexible}.

We specified the state-dependent distributions for step lengths and turning angles as
$$
\text{step}_t \mid \{S_t = j\} \sim \text{Gamma}(\mu_j, \sigma_j), \quad
\text{angle}_t \mid \{S_t = j\} \sim \text{wrapped Cauchy}(\nu_j, \rho_j),
$$
for states $j = 1, \dots, N$. After some initial exploration we decided that a $2$-state model --- distinguishing between predominantly inactive and more active behaviour ---
was the most appropriate. 
Although a 3-state model further subdivided the movement state, the overlap between the state-dependent distributions was substantial, and the estimated transition probability matrix exhibited little persistence. This indicates that, at an hourly temporal resolution, different types of active movement could not be reliably disentangled.
We modelled the state-switching dynamics of the latent $2$-state Markov chain as 
$$
\text{logit}(\gamma_{12}^{(t)}) = \beta_0^{(12)} + h_{12}(\text{hour}_t), \quad \text{logit}(\gamma_{21}^{(t)}) = \beta_0^{(21)} + h_{21}(\text{hour}_t) + u(\bm{r}_t),
$$
where $\bm{r}_t$ is the 2D location in the study area at time $t$ and $u$ is a Matérn field. The periodic function in both equations is
$$h_{ij}(\text{hour}) = \sum_{k=1}^3 \beta_{1k}^{(ij)} \sin \Bigl(\frac{2 \pi k \, \text{hour}}{24} \Bigl) + \sum_{k=1}^3 \beta_{2k}^{(ij)} \cos \Bigl(\frac{2 \pi k \, \text{hour}}{24} \Bigr).$$
Only the linear predictor for the transition from state~2 (active) to state~1 (resting) includes the spatial field, as state~1 is characterised by virtually no movement. Consequently, a switch to state~2 occurs very near to the location where the animal was resting and depends primarily on the time spent there and the time of day.

\begin{figure}
    \centering
    \includegraphics[width=1\linewidth]{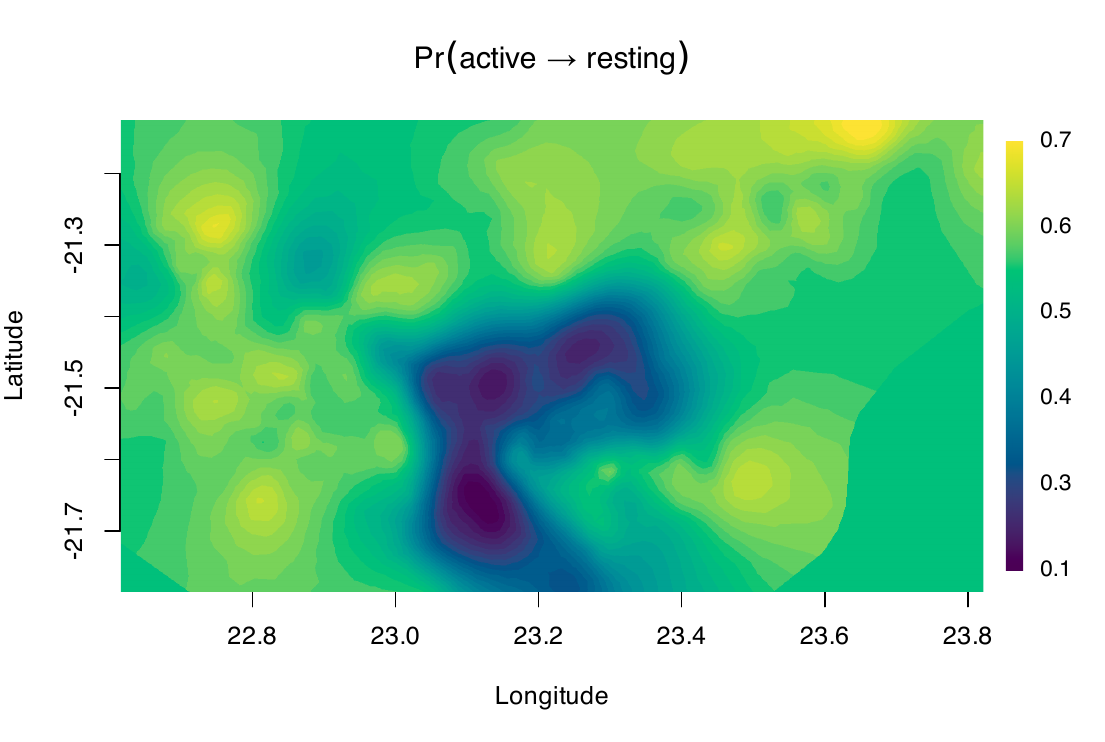}
    \caption{Transition probability from active to resting, based on the posterior mode of the spatial field fitted to the lion data. Dark blue areas indicate a low probability of transitioning into the resting state, while yellow areas indicate a high probability. Pixels in the lower-right corner fall outside the triangulated mesh, so the field is predicted as zero there.}
    \label{fig:spatial_field}
\end{figure}

Again using \texttt{fmesher}, we approximated the spatial field by constructing a triangulated mesh with $2516$ nodes, leading to GMRF approximating with $l = 2516$ weights and precision matrix as given in Equation \eqref{eqn:SPDE_precision}. The mesh is shown in Appendix~\ref{a5:figs}.
Fitting the model with a bandwidth parameter of $k = 15$ took about $6$ minutes on an Apple M2 chip with 16GB of memory. Additionally, we also fitted a homogeneous model and models including only the random field or only the periodic functions $f_{ij}$.
For the fully parameteric models, fitting times where negligible in comparison to the models including a spatial field. Fitting the purely spatial model took about $4$ minutes.


\begin{figure}
    \centering
    \includegraphics[width=0.85\linewidth]{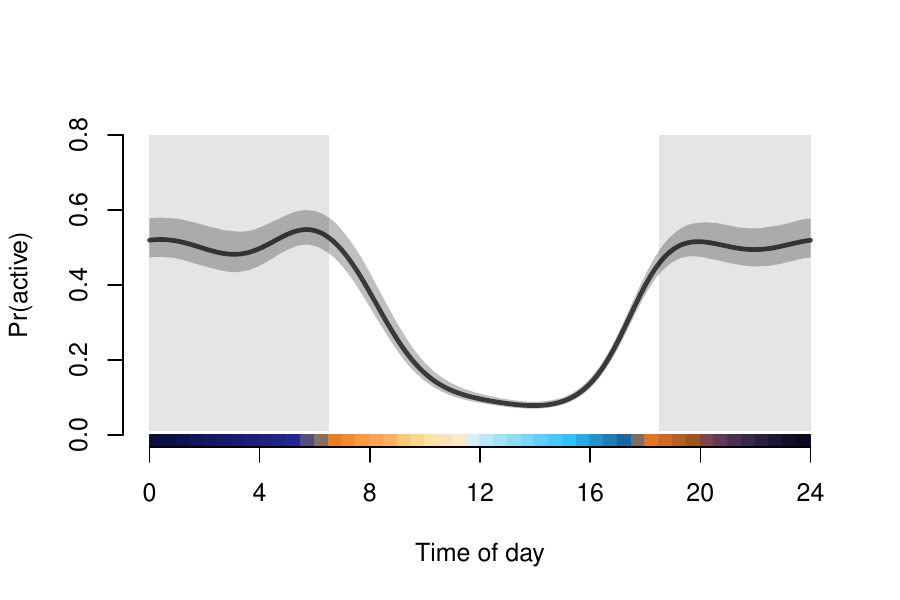}
    \caption{Probability of being active as a function of the time of day, obtained based on the periodically stationary state distribution. The gray band corresponds to pointwise 95\% confidence intervals based on the approximate multivariate normal distribution of the MLE.}
    \label{fig:pstationary}
\end{figure}


We find that the animals are very unlikely to transition into the resting state within a distinct region in the centre of the study area (cf.\ Figure~\ref{fig:spatial_field}). In other words, once active in this area, they tend to remain active rather than settle. A comparison with satellite imagery did not reveal any obvious landscape features explaining this pattern. This spatial structure may therefore reflect more subtle ecological drivers, such as prey distribution, social interactions, or unobserved habitat characteristics, and warrants further ecological investigation.

Isolating the effect of time of day via the periodically stationary distribution \citep{koslik2025inference} shows a clear diel pattern (cf. Figure~\ref{fig:pstationary}). Activity is highest during the night and early morning hours, and lowest during the daytime, which is consistent with the predominantly nocturnal behaviour of lions in arid environments.

\section{Discussion}

We consider combining hidden Markov models (HMMs) and latent Gaussian fields in a unified modelling framework, where, depending on the application, the GF component may enter at different levels of the model hierarchy.
Recent advances in automatic differentiation and automatic Laplace approximation have, in principle, made such models computationally feasible. In particular, these tools integrate naturally with the SPDE approach, which provides Gaussian Markov random field approximations to Gaussian fields and consequently yields sparse precision matrices. This sparsity is essential for scalable inference when the latent Gaussian component is high-dimensional.
A central computational obstacle specific to HMMs is that marginalisation over the discrete latent states --- required to evaluate the HMM log-likelihood --- induces global temporal dependence. Consequently, the Hessian matrix with respect to the latent variables is dense, undermining the efficiency gains of the automatic Laplace approximation as implemented in (\texttt{R})\texttt{TMB}.

To address this issue, we developed a modified forward algorithm for evaluating the HMM log-likelihood. The key feature of this construction is that it generates a sparse Hessian, thereby restoring compatibility with SPDE-based GMRF representations and automatic Laplace approximation tools. 
Taken together, these three ingredients allow rapid and scalable inference in models that integrate discrete state switching with richly parameterised Gaussian fields. 

We illustrated the practical performance of the approach in two case studies. The first concerned stellar flare detection, where a state-switching flaring process was masked by quasi-periodic oscillations in stellar brightness modelled via a Gaussian process. The second analysed movement data from African lions in the Central Kalahari, where behavioural switching was modulated by a latent spatial field entering the HMM likelihood. In both cases, the proposed framework enabled flexible modelling while retaining remarkable computational efficiency.


Several extensions appear promising. For instance, the SPDE approach has already been generalised to latent fields that are non-stationary \citep{lindgren2011explicit}, and non-Gaussian or spatio-temporal processes are also possible \citep{lindgren2022spde}. 
Furthermore, the SPDE framework can also be applied to non-Euclidean manifolds \citep{lindgren2011explicit}, such as curved surfaces or even networks. Exploring these richer settings as building blocks inside HMMs is an interesting avenue for future research. 

\section*{Acknowledgments}

The author thanks Roland Langrock for helpful comments on an earlier version of this manuscript as well as Kasper Kristensen for answering various questions related to \texttt{RTMB}.

\section*{Conflict of Interest}
The author declares that there are no conflicts of interest.

\printbibliography

@article{lindgren2011explicit,
  title={An explicit link between Gaussian fields and Gaussian Markov random fields: the stochastic partial differential equation approach},
  author={Lindgren, Finn and Rue, H{\aa}vard and Lindstr{\"o}m, Johan},
  journal={Journal of the Royal Statistical Society Series B: Statistical Methodology},
  volume={73},
  number={4},
  pages={423--498},
  year={2011},
  publisher={Oxford University Press}
}

@article{lystig2002exact,
  title={Exact computation of the observed information matrix for hidden {M}arkov models},
  author={Lystig, Theodore C and Hughes, James P},
  journal={Journal of {C}omputational and {G}raphical {S}tatistics},
  volume={11},
  number={3},
  pages={678--689},
  year={2002},
  publisher={Taylor \& Francis}
}

@article{le2000exponential,
  title={Exponential forgetting and geometric ergodicity in hidden Markov models},
  author={Le Gland, Franc{\c{c}}ois and Mevel, Laurent},
  journal={Mathematics of Control, Signals and Systems},
  volume={13},
  number={1},
  pages={63--93},
  year={2000},
  publisher={Springer}
}

@book{zucchini2016hidden,
  title={Hidden {M}arkov {M}odels for {T}ime {S}eries: {A}n {I}ntroduction using {R}},
  author={Zucchini, Walter and MacDonald, Iain L and Langrock, Roland},
  year={2016},
  publisher={Chapman and Hall/ CRC press},
  edition={2nd}
}

@article{erkanli1994laplace,
  title={Laplace approximations for posterior expectations when the mode occurs at the boundary of the parameter space},
  author={Erkanli, Alaattin},
  journal={Journal of the {A}merican {S}tatistical {A}ssociation},
  volume={89},
  number={425},
  pages={250--258},
  year={1994},
  publisher={Taylor \& Francis}
}

@book{van2000asymptotic,
  title={Asymptotic {S}tatistics},
  author={Van der Vaart, Aad W},
  volume={3},
  year={2000},
  publisher={Cambridge {U}niversity {P}ress}
}

@article{langrock2015nonparametric,
  title={Nonparametric inference in hidden {M}arkov models using {P}-splines},
  author={Langrock, Roland and Kneib, Thomas and Sohn, Alexander and DeRuiter, Stacy L},
  journal={{B}iometrics},
  volume={71},
  number={2},
  pages={520--528},
  year={2015},
  publisher={Oxford University Press}
}

@article{langrock2017markov,
  title={Markov-switching generalized additive models},
  author={Langrock, Roland and Kneib, Thomas and Glennie, Richard and Michelot, Th{\'e}o},
  journal={Statistics and {C}omputing},
  volume={27},
  pages={259--270},
  year={2017},
  publisher={Springer}
}

@article{chen2023bayesian,
  title={Bayesian spline-based hidden {M}arkov models with applications to actimetry data and sleep analysis},
  author={Chen, Sida and Finkenst{\"a}dt, B{\"a}rbel},
  journal={Journal of the {A}merican {S}tatistical {A}ssociation},
  pages={1--11},
  year={2023},
  publisher={Taylor \& Francis}
}

@article{feldmann2023flexible,
  title={Flexible modelling of diel and other periodic variation in hidden {M}arkov models},
  author={Feldmann, Carlina C and Mews, Sina and Coculla, Angelica and Stanewsky, Ralf and Langrock, Roland},
  journal={Journal of {S}tatistical {T}heory and {P}ractice},
  volume={17},
  number={3},
  pages={45},
  year={2023},
  publisher={Springer}
}

@article{michelot2025hmmtmb,
  title={hmmTMB: {H}idden {M}arkov models with flexible covariate effects in {R}},
  author={Michelot, Th{\'e}o},
  journal={Journal of {S}tatistical {S}oftware},
  volume={114},
  pages={1--45},
  year={2025}
}

@article{koslik2024efficient,
  title={Efficient smoothness selection for nonparametric {M}arkov-switching models via quasi restricted maximum likelihood},
  author={Koslik, Jan-Ole},
  journal={arXiv preprint arXiv:2411.11498},
  year={2024}
}

@article{koslik2025tensor,
  title={Tensor-product interactions in {M}arkov-switching models},
  author={Koslik, Jan-Ole},
  journal={arXiv preprint arXiv:2507.01555},
  year={2025}
}

@article{michels2025nonparametric,
   author = {Rouven Michels and Roland Langrock},
   title ={Nonparametric estimation of bivariate hidden {M}arkov models using tensor-product {B}-splines},
   journal = {Statistical {M}odelling},
   volume = {0},
   number = {0},
   pages = {1471082X251335431},
   year = {2025}
}

@article{koslik2025flexible,
  title={Flexible unimodal density estimation in hidden {M}arkov models},
  author={Koslik, Jan-Ole and Dupont, Fanny and Auger-M{\'e}th{\'e}, Marie and Marcoux, Marianne and Hussey, Nigel and Heckman, Nancy},
  journal={arXiv preprint arXiv:2511.17071},
  year={2025}
}

@article{ammann2026non,
  title={Non-homogeneous {M}arkov-switching generalised additive models for location, scale, and shape},
  author={Ammann, Katharina and Adam, Timo and Koslik, Jan-Ole},
  journal={arXiv preprint arXiv:2601.03760},
  year={2026}
}

@article{ohagan1978curve,
  title={Curve fitting and optimal design for prediction},
  author={O'Hagan, Anthony},
  journal={Journal of the {R}oyal {S}tatistical {S}ociety: {S}eries {B} ({M}ethodological)},
  volume={40},
  number={1},
  pages={1--24},
  year={1978},
  publisher={Wiley Online Library}
}

@article{foreman2017fast,
  title={Fast and scalable {G}aussian process modeling with applications to astronomical time series},
  author={Foreman-Mackey, Daniel and Agol, Eric and Ambikasaran, Sivaram and Angus, Ruth},
  journal={The {A}stronomical {J}ournal},
  volume={154},
  number={6},
  pages={220},
  year={2017},
  publisher={IOP Publishing}
}

@article{esquivel2025detecting,
  title={Detecting stellar flares in photometric data using hidden {M}arkov models},
  author={Esquivel, J Arturo and Shen, Yunyi and Leos-Barajas, Vianey and Eadie, Gwendolyn and Speagle, Joshua S and Craiu, Radu V and Medina, Amber and Davenport, James RA},
  journal={The {A}strophysical {J}ournal},
  volume={979},
  number={2},
  pages={141},
  year={2025},
  publisher={IOP Publishing}
}

@article{kristensen2016tmb,
  title={{TMB}: automatic differentiation and {L}aplace approximation},
  author={Kristensen, Kasper and Nielsen, Anders and Berg, Casper W and Skaug, Hans and Bell, Bradley M},
  journal={Journal of {S}tatistical {S}oftware},
  volume={70},
  pages={1--21},
  year={2016}
}

@Manual{kristensen2024rtmb,
    title = {{RTMB}: `R' Bindings for {TMB}'},
    author = {Kasper Kristensen},
    year = {2026},
    note = {R package version 1.8},
    url = {https://CRAN.R-project.org/package=RTMB},
}

@Manual{koslikLaMa2025,
  title = {{L}a{M}a: {F}ast {N}umerical {M}aximum {L}ikelihood {E}stimation for {L}atent {M}arkov {M}odels},
  author = {Jan-Ole Fischer},
  year = {2026},
  note = {{R} package version 2.1.0},
  url = {https://CRAN.R-project.org/package=LaMa},
}

@article{mews2025build,
  title={How to build your latent {M}arkov model: {T}he role of time and space},
  author={Mews, Sina and Koslik, Jan-Ole and Langrock, Roland},
  journal={Statistical {M}odelling},
note = {in press},
  publisher={SAGE Publications Sage India: New Delhi, India},
  year={2025}
}

@article{mcclintock2020uncovering,
  title={Uncovering ecological state dynamics with hidden {M}arkov models},
  author={McClintock, Brett T and Langrock, Roland and Gimenez, Olivier and Cam, Emmanuelle and Borchers, David L and Glennie, Richard and Patterson, Toby A},
  journal={{E}cology {L}etters},
  volume={23},
  number={12},
  pages={1878--1903},
  year={2020},
  publisher={Wiley Online Library}
}

@article{mews2022multistate,
  title={Multistate capture--recapture models for irregularly sampled data},
  author={Mews, Sina and Langrock, Roland and King, Ruth and Quick, Nicola},
  journal={The Annals of Applied Statistics},
  volume={16},
  number={2},
  pages={982--998},
  year={2022},
  publisher={Institute of Mathematical Statistics}
}

@article{liu2012stock,
  title={Stock market volatility and equity returns: {E}vidence from a two-state {M}arkov-switching model with regressors},
  author={Liu, Xinyi and Margaritis, Dimitris and Wang, Peiming},
  journal={Journal of {E}mpirical {F}inance},
  volume={19},
  number={4},
  pages={483--496},
  year={2012},
  publisher={Elsevier}
}

@article{zhang2019high,
  title={High-order hidden {M}arkov model for trend prediction in financial time series},
  author={Zhang, Mengqi and Jiang, Xin and Fang, Zehua and Zeng, Yue and Xu, Ke},
  journal={Physica {A}: {S}tatistical {M}echanics and its {A}pplications},
  volume={517},
  pages={1--12},
  year={2019},
  publisher={Elsevier}
}

@article{amoros2019continuous,
  title={A continuous-time hidden {M}arkov model for cancer surveillance using serum biomarkers with application to hepatocellular carcinoma},
  author={Amoros, Ruben and King, Ruth and Toyoda, Hidenori and Kumada, Takashi and Johnson, Philip J and Bird, Thomas G},
  journal={Metron},
  volume={77},
  number={2},
  pages={67--86},
  year={2019},
  publisher={Springer}
}

@article{soper2020hidden,
  title={A hidden {M}arkov model for population-level cervical cancer screening data},
  author={Soper, Braden C and Nyg{\aa}rd, Mari and Abdulla, Ghaleb and Meng, Rui and Nyg{\aa}rd, Jan F},
  journal={{S}tatistics in {M}edicine},
  volume={39},
  number={25},
  pages={3569--3590},
  year={2020},
  publisher={Wiley Online Library}
}

@Manual{lindgren2025fmesher,
    title = {fmesher: {T}riangle {M}eshes and {R}elated {G}eometry {T}ools},
    author = {Finn {L}indgren},
    year = {2026},
    note = {{R} package version 0.7.0},
    url = {https://CRAN.R-project.org/package=fmesher}
  }

@phdthesis{MacFarlane2014,
  author    = {MacFarlane, K.},
  title     = {The ecology and management of {K}alahari lions in a conflict area in central {B}otswana},
  school    = {{A}ustralian {N}ational {U}niversity},
  address   = {Canberra (Australia)},
  year      = {2014},
  type      = {Dissertation}
}

@article{koslik2025inference,
  title={Inference on the state process of periodically inhomogeneous hidden {M}arkov models for animal behavior},
  author={Koslik, Jan-Ole and Feldmann, Carlina C and Mews, Sina and Michels, Rouven and Langrock, Roland},
  journal={The {A}nnals of {A}pplied {S}tatistics},
  volume={19},
  number={4},
  pages={2724--2737},
  year={2025},
  publisher={Institute of Mathematical Statistics}
}

@article{langrock2012flexible,
  title={Flexible and practical modeling of animal telemetry data: hidden {M}arkov models and extensions},
  author={Langrock, Roland and King, Ruth and Matthiopoulos, Jason and Thomas, Len and Fortin, Daniel and Morales, Juan M},
  journal={{E}cology},
  volume={93},
  number={11},
  pages={2336--2342},
  year={2012},
  publisher={Wiley Online Library}
}

@article{bolin2009wavelet,
  title={Wavelet {M}arkov models as efficient alternatives to tapering and convolution fields},
  author={Bolin, David and Lindgren, Finn},
  journal={Preprints in {M}athematical {S}ciences},
  volume={2009},
  number={13},
  year={2009},
  publisher={Lund University}
}

@book{cressie2015statistics,
  title={Statistics for {S}patial {D}ata},
  author={Cressie, Noel},
  year={2015},
  publisher={John Wiley \& Sons}
}

@article{michels2026integrating,
  title={Integrating unsupervised and supervised learning for the prediction of defensive schemes in {A}merican football},
  author={Michels, Rouven and Bajons, Robert and Fischer, Jan-Ole},
  journal={arXiv preprint arXiv:2602.10784},
  year={2026}
}

@article{tierney1986accurate,
  title={Accurate approximations for posterior moments and marginal densities},
  author={Tierney, Luke and Kadane, Joseph B},
  journal={Journal of the {A}merican {S}tatistical {A}ssociation},
  volume={81},
  number={393},
  pages={82--86},
  year={1986},
  publisher={Taylor \& Francis}
}

@article{ogden2021error,
  title={On the error in {L}aplace approximations of high-dimensional integrals},
  author={Ogden, Helen},
  journal={Stat},
  volume={10},
  number={1},
  pages={e380},
  year={2021},
  publisher={Wiley Online Library}
}

@article{bachl2019inlabru,
  title={inlabru: an {R} package for {B}ayesian spatial modelling from ecological survey data},
  author={Bachl, Fabian E and Lindgren, Finn and Borchers, David L and Illian, Janine B},
  journal={Methods in {E}cology and {E}volution},
  volume={10},
  number={6},
  pages={760--766},
  year={2019},
  publisher={Wiley Online Library}
}

@article{lindgren2015bayesian,
  title={Bayesian spatial modelling with {R-INLA}},
  author={Lindgren, Finn and Rue, H{\aa}vard},
  journal={Journal of {S}tatistical {S}oftware},
  volume={63},
  pages={1--25},
  year={2015}
}

@article{otting2023copula,
  title={A copula-based multivariate hidden {M}arkov model for modelling momentum in football},
  author={{\"O}tting, Marius and Langrock, Roland and Maruotti, Antonello},
  journal={{AStA} {A}dvances in {S}tatistical {A}nalysis},
  volume={107},
  number={1},
  pages={9--27},
  year={2023},
  publisher={Springer}
}

@article{lindgren2022spde,
  title={The {SPDE} approach for {G}aussian and non-{G}aussian fields: 10 years and still running},
  author={Lindgren, Finn and Bolin, David and Rue, H{\aa}vard},
  journal={Spatial {S}tatistics},
  volume={50},
  pages={100599},
  year={2022},
  publisher={Elsevier}
}

@article{viterbi2003error,
  title={Error bounds for convolutional codes and an asymptotically optimum decoding algorithm},
  author={Viterbi, Andrew},
  journal={{IEEE} {T}ransactions on {I}nformation {T}heory},
  volume={13},
  number={2},
  pages={260--269},
  year={2003},
  publisher={IEEE}
}

@book{bates1988nonlinear,
  title={Nonlinear {R}egression {A}nalysis and its {A}pplications},
  author={Bates, Douglas M and Watts, Donald G},
  volume={2},
  year={1988},
  publisher={Wiley New York}
}

@article{mcclintock2012general,
  title={A general discrete-time modeling framework for animal movement using multistate random walks},
  author={McClintock, Brett T and King, Ruth and Thomas, Len and Matthiopoulos, Jason and McConnell, Bernie J and Morales, Juan M},
  journal={{E}cological {M}onographs},
  volume={82},
  number={3},
  pages={335--349},
  year={2012},
  publisher={Wiley Online Library}
}

\appendix
\renewcommand{\thesubsection}{\Alph{section}.\arabic{subsection}}

\renewcommand{\thefigure}{A.\arabic{figure}}
\setcounter{figure}{0}

\section{Appendix}

\subsection{Proofs}
\label{a1:proof}

\subsubsection*{Notation}

Let $\ell_T(\bm{\theta}, \bm{x})$ denote the log-likelihood function based on $T$ observations. This can either be the log-likelihood of latent observations $x_1, \dots, x_T$ ($\log f(x_1, \dots, x_T)$) or the log-likelihood of observations $y_1, \dots, y_T$ given a latent field ($\log f(y_1, \dots, y_T \mid \bm{x})$). 
Let $\tilde\ell_T(\bm{\theta}, \bm{x})$ be the respective approximation given in Equation \eqref{eqn:log_likelihood_approx}.

Let $\bm{\phi}_t$ denote the exact scaled forward variable with entries $\phi_t(j) = \Pr(S_t = j \mid x_1, \dots, x_t)$,
and let $\tilde{\bm{\phi}}_t^k$ denote the approximation based on the last $k$ observations, with entries
$\tilde{\phi}_t^k(j) = \Pr(S_t = j \mid x_{t-k+1}, \dots, x_t)$.
Now define the prediction filter $\bm{p}_t = \bm \phi_{t-1} \bm{\Gamma}^{(t)}$ and its approximation $\tilde{\bm{p}}_t^k = \tilde{\bm\phi}_{t-1} \bm{\Gamma}^{(t)}$.

\subsubsection*{Assumptions}
We make the following assumptions.
\begin{enumerate}
    \item \textbf{Compact parameter and latent space:} 
    $\bm{\Theta} \times \mathcal{X}$ is a compact subset of $\mathbb{R}^p \times \mathbb{R}^d$.
    
    \item \textbf{Bounded state-dependent densities:} 
    For all states $j$ and all times $t$, the emission densities satisfy
    $$
    0 < m \le f_j(\cdot) \le M < \infty,
    $$
    uniformly over all relevant $(\bm{\theta},\bm{x}) \in \bm{\Theta} \times \mathcal{X}$.
    
    \item \textbf{Uniform ergodicity of the Markov chain:} 
    There exists a block length $l \ge 1$ and a constant $\varepsilon > 0$ such that for all $t$,
    $$
    \big(\bm{\Gamma}^{(t)} \bm{\Gamma}^{(t+1)} \cdots \bm{\Gamma}^{(t+l-1)}\big)_{ij} \ge \varepsilon
    \quad \text{for all states } i,j.
    $$
    This ensures that the Markov chain is uniformly ergodic. We call $l$ the index of primitivity of the chain.
\end{enumerate}

\begin{lemma}[Filter bound propagates to log-likelihood]
\label{lem:filter_to_likelihood}
Under Assumptions (1)-(3), we have
$$
\sup_{(\bm{\theta},\bm{x}) \in \bm{\Theta} \times \mathcal{X}} 
\big| \ell_T(\bm{\theta},\bm{x}) - \tilde{\ell}_T(\bm{\theta},\bm{x}) \big|
\le \sum_{t=1}^T \frac{M}{m} \sup_{(\bm{\theta}, \bm{x})} \|\bm{p}_t - \tilde{\bm{p}}_t^k\|
$$
where $\| \bm u\| = \sum_i \vert u_i\vert$ denotes the $L_1$-norm of a vector $\bm u$. 
If a random sequence $y_1, \dots, y_T$ is involved, the inequality holds pathwise (i.e., for each realization of $\bm y$).
\end{lemma}
\begin{proof}
For $j = 1, \dots, N$, let $f_j^{(t)}$ denote the \emph{evaluated} state-dependent densities, i.e.\ $f_j(x_t \mid \bm x; \bm \theta)$ (case 1) or $f_j(y_t \mid \bm x; \bm \theta)$ (case 2). Then the conditional densities multiplied in the likelihood are
$$
f_t = \sum_{i=1}^N \sum_{j=1}^N \phi_{t-1}(i) \gamma_{ij}^{(t)} f_j^{(t)} = \sum_{j=1}^N p_t(j) f_j^{(t)}.
$$
Denote with $\tilde{f}_t^k$ the same quantity with $\bm{\phi}_{t-1}$ replaced by $\tilde{\bm{\phi}}_{t-1}^k$.
Since the $f_j^{(t)}$ are bounded, the error in the conditional densities is also bounded:
$$
\sup_{(\bm{\theta}, \bm{x})} |f_t - \tilde{f}_t^k| \le M \sup_{(\bm{\theta}, \bm{x})} \|\bm{p}_{t} - \tilde{\bm{p}}_{t}^k\|.
$$

Applying the mean value theorem for logarithms gives
$$
|\log f_t - \log \tilde{f}_t^k| = \frac{|f_t - \tilde{f}_t^k|}{\xi},
$$
for some $\xi \in (\min(f_t, \tilde{f}_t^k), \max(f_t, \tilde{f}_t^k)) \ge m$. Hence,
$$
\sup_{(\bm{\theta}, \bm{x})} |\log f_t - \log \tilde{f}_t^k| \le \frac{M}{m} \sup_{(\bm{\theta}, \bm{x})} \|\bm{p}_{t} - \tilde{\bm{p}}_{t}^k\|.
$$
Summing over $t = 1, \dots, T$, and using the triangle inequality, we obtain
$$
\sup_{(\bm{\theta}, \bm{x})} |\ell_T(\bm{\theta}, \bm{x}) - \tilde{\ell}_T(\bm{\theta}, \bm{x})|
\le \sum_{t=1}^T \frac{M}{m} \sup_{(\bm{\theta}, \bm{x})} \|\bm{p}_{t} - \tilde{\bm{p}}_{t}^k\|.
$$
\end{proof}

\begin{proof}[Proof of Theorem 1: Case 1]
We first cover the deterministic result, i.e.\ when 
$\ell_T(\bm{x}, \bm\theta) = \log f_\theta(x_1, \dots, x_T)$, purely a function of a parameter $\bm\theta$ and latent variables $\bm{x}$. We show that the approximation error decays geometrically in $k$, that is
$$
\sup_{(\bm{\theta},\bm{x}) \in \bm{\Theta} \times \mathcal{X}} \|\bm{p}_t - \tilde{\bm{p}}_t^k\| \le C \rho^k,
$$
$C > 0$ and $\rho \in (0,1)$, relying on the work of \citet{le2000exponential}. 
Under our assumptions, the assumptions of \citet{le2000exponential} are satisfied.
The truncated filter $\tilde{\bm p}_t^k$ corresponds to a filter started at time $t-k$ from an arbitrary initial distribution and propagated forward to time $t$. In the notation of \citet{le2000exponential}, this is equivalent to comparing filters on an interval $[m,n]$ with $m = t-k$ and $n = t$, so that the exponential forgetting rate depends only on the lag $n-m = k$. In our notation, Theorem 2.1 states
$$
\|\bm{p}_t - \tilde{\bm{p}}_t^k\| \leq 2 \prod_{\kappa=1}^{\lfloor (k+1)/l \rfloor} (1 - \epsilon^l[\delta(x_{(t-k) + (\kappa-1)r+1}) \cdots \delta(x_{(t-k)+\kappa r - 1})]),
$$
where $l$ is the index of primitivity of the chain (see Assumptions), $\delta(x) = \frac{\max_j f_j(x)}{\min_j f_j(x)}$ $\epsilon = \text{min}_+ \gamma_{ij}$ is the minimum of the positive entries of the t.p.m. Since in our case $\bm \Gamma^{(t)}$ might be time-varying, we use $\epsilon = \underset{i,j,t}{\text{min}_+} \gamma^{(t)}_{ij}$. The product over $\kappa$ corresponds to successive blocks of length $l$.
Boundedness of the state-dependent densities implies $\delta(x) \leq \frac{M}{m}$ for any $x$. Hence we arrive at the more crude bound
$$
\|\bm{p}_t - \tilde{\bm{p}}_t^k\| \leq 2 \prod_{\kappa=1}^{\lfloor (k+1)/l \rfloor} (1 - (\epsilon \tfrac{m}{M})^l) = 2 (1 - \epsilon [\tfrac{m}{M}]^l)^{\lfloor (k+1)/l \rfloor} \le 2 \bigl((1 - (\epsilon \tfrac{m}{M})^l)^{1/l} \bigr)^k = C \rho^k,
$$
and application of Lemma~\ref{lem:filter_to_likelihood} gives
$$
\sup_{(\bm{\theta},\bm{x}) \in \bm{\Theta} \times \mathcal{X}} 
\big| \ell_T(\bm{\theta},\bm{x}) - \tilde{\ell}_T(\bm{\theta},\bm{x}) \big|
\le \sum_{t=1}^T \frac{M}{m} \sup_{(\bm{\theta}, \bm{x})} \|\bm{p}_t - \tilde{\bm{p}}_t^k\| 
\le T \frac{MC}{m} \rho^k.
$$
\end{proof}

\begin{proof}[Proof of Theorem 1: Case 2]
We now cover the case where $\ell_T(\bm{x}, \bm\theta) = \log f_\theta(y_1, \dots, y_T \mid \bm x)$, is a function of a parameter $\bm\theta$, latent variables $\bm{x}$ and a stochastic process generated from the model. Hence, a deterministic result is not possible and we show
$$
\sup_{(\bm{\theta}, \bm{x}) \in \bm{\Theta} \times \mathcal{X}} \left| \ell_T(\bm{\theta}, \bm{x}) - \tilde{\ell}_T(\bm{\theta}, \bm{x}) \right| \le T C \rho^k \quad \text{a.s.},
$$
In this case, we can indeed directly apply Theorem 2.2 of \citet{le2000exponential} and have 
$$
\sup_{(\bm{\theta}, \bm{x})} \|\bm{p}_t - \tilde{\bm{p}}_t^k\| \le C \rho^k \quad \text{a.s.}, 
$$
where $\rho$ again depends on the index of primitivity of the chain and the minimum positive entry of the t.p.m. 
Pathwise application of Lemma~\ref{lem:filter_to_likelihood} then again gives
$$
\sup_{(\bm{\theta},\bm{x}) \in \bm{\Theta} \times \mathcal{X}} 
\big| \ell_T(\bm{\theta},\bm{x}) - \tilde{\ell}_T(\bm{\theta},\bm{x}) \big|
\le \sum_{t=1}^T \frac{M}{m} \sup_{(\bm{\theta}, \bm{x})} \|\bm{p}_t - \tilde{\bm{p}}_t^k\| 
\le T \frac{MC}{m} \rho^k, \quad \text{a.s.}
$$
\end{proof}

\subsubsection*{Remarks on the assumptions}

\begin{remark}[Compactness of $\mathcal{X}$]
Within the Laplace approximation, the objective
$$
\ell_T(\bm{\theta}, \bm{x}) - \frac{1}{2} \bm{x}^\top \bm{Q} \bm{x}
$$
is maximised at $\hat{\bm{x}}_\theta$. The quadratic term dominates as $\|\bm{x}\| \to \infty$, allowing restriction to a compact ball around $\hat{\bm{x}}_\theta$.
\end{remark}

\begin{remark}[Boundedness of state-dependent densities]
As in the main manuscript, notation is complicated by the fact that we either evaluate the forward algorithm on latent variables $x_1, \dots, x_t$ (case 1) or genuine observations $y_1, \dots, y_T$ (case 2). However, under Assumption (1) and the additional assumption of non-zero variance parameters, the state-dependent densities are bounded in any of the following cases:
\begin{enumerate}
    \item $f_j(y_t; \bm{\theta})$, independent of $\bm{x}$ and assuming non-zero variance parameters.
    \item $f_j(y_t; \bm{\theta}, \bm{x})$, non-zero variance and with $\bm{x}$ in compact set $\mathcal{X}$.
    \item $f_j(x_t; \bm{\theta}, \bm{x})$, non-zero variance and with $\bm{x}$ in compact set $\mathcal{X}$.
\end{enumerate}
\end{remark}

\begin{remark}[Uniform ergodicity of the Markov chain]
If all entries of $\bm\Gamma^{(t)}$ are positive and depend on bounded covariates and parameters in $\bm\theta$, or on random effects $\bm x$, the compactness of $\bm{\Theta} \times \mathcal{X}$ implies $\gamma_{ij}^{(t)} \ge \varepsilon > 0$.

Ergodicity is slightly weaker but also holds if entries are fixed at $\gamma_{ij}^{(t)} = 0$, as long as there exist an $l$, such that $(\bm\Gamma^{(t)} \bm\Gamma^{(t+1)} \cdots \bm\Gamma^{(t+l)})_{ij} \ge \varepsilon > 0$ for all $i,j,t$. The integer $l$ is called the index of primitivity of the chain, which is greater than 1 in this case. 
Consequently, the assumption is violated if the Markov chain is periodic or reducible and in such a case the forward-approximation is not applicable.
\end{remark}

\subsection{Quality of the Laplace approximation}
\label{a2:laplace}

The Laplace approximation replaces the conditional distribution of the latent variables $\bm x$ given the data with a Gaussian distribution centered at its mode \citep{tierney1986accurate}.
If the function $g$ in the exponent of the posterior is exactly quadratic, the Laplace approximation is exact. This occurs, for example, when both the observation and process distributions are Gaussian and linear in $\bm{x}$. In our model, this condition only holds for one of the two densities, while the other corresponds to a non-Gaussian, non-linear HMM likelihood.

Even in this case, the Laplace approximation provides a local Gaussian approximation around the posterior mode. Its accuracy depends on the posterior being dominated by a single mode. In practice, it often performs well, particularly when the posterior is concentrated around a single mode and the sample size is sufficient to constrain the distribution in a relatively small region of the high-dimensional space, i.e., when the posterior is sharply peaked near the mode, corresponding to relatively large eigenvalues of the Hessian $\bm{H}_\theta$ \citep{tierney1986accurate, ogden2021error}.

\subsection{A simulation experiment}
\label{a3:sim}

All code for full reproducibility of the simulation experiment is available at \url{https://github.com/janolefi/HMMs_GFs}.

We generate data from a 2-state HMM where the transition probability matrix is driven by a spatial field. We simulate state-dependent gamma-distributed step lengths, reparameterised in terms of mean and standard deviation, 
$$
\text{step}_t \mid \{S_t = j\} \sim \text{Gamma}(\mu_j, \sigma_j), \quad{j = 1,2}
$$
with mean values $\mu_1 = 0.2$, $\mu_2 = 5$ and standard deviations $\sigma_1 = 0.5$, $\sigma_2 = 3$.
At each time point, we also randomly draw a turning angle from a von Mises distribution, independently of the latent state. The mean is chosen as to imply a biased random walk towards the location $(0,0)$ to prevent the simulated track from dispersing too much, with the concentration parameter being set to $0.3$.
We then compute Euclidean locations at each time point, initialising with the position $(0,0)$ and calculating the subsequent position based on the sampled step length and direction.
The transition probability $\gamma_{12}$ is homogeneous while $\gamma_{21}^{(t)}$ depends on a spatial field. 
Specifically, we set
$$\text{logit}(\gamma_{12}) = \beta_0^{(12)} \quad \text{and} \quad \text{logit}(\gamma_{21}^{(t)}) = \beta_0^{(21)} + u(\bm{r}_t),$$
with $\beta_0^{(12)} = \beta_0^{(21)} = -2$. We fix the intial distribution of the chain at $\bm{\delta}^{(1)} = (0.5,0.5)$ and do not estimate it as it is not the focus of this study and usually of minor practical relevance.
The true spatial field $u$ which we use for simulation is parametric with 
$$
u(\bm{r}) = u(x,y) = 2 \sin(2 \pi \tfrac{x}{40}) + 2 \cos(2 \pi \tfrac{y}{40}).
$$

\begin{figure}
    \centering
    \includegraphics[width=1\linewidth]{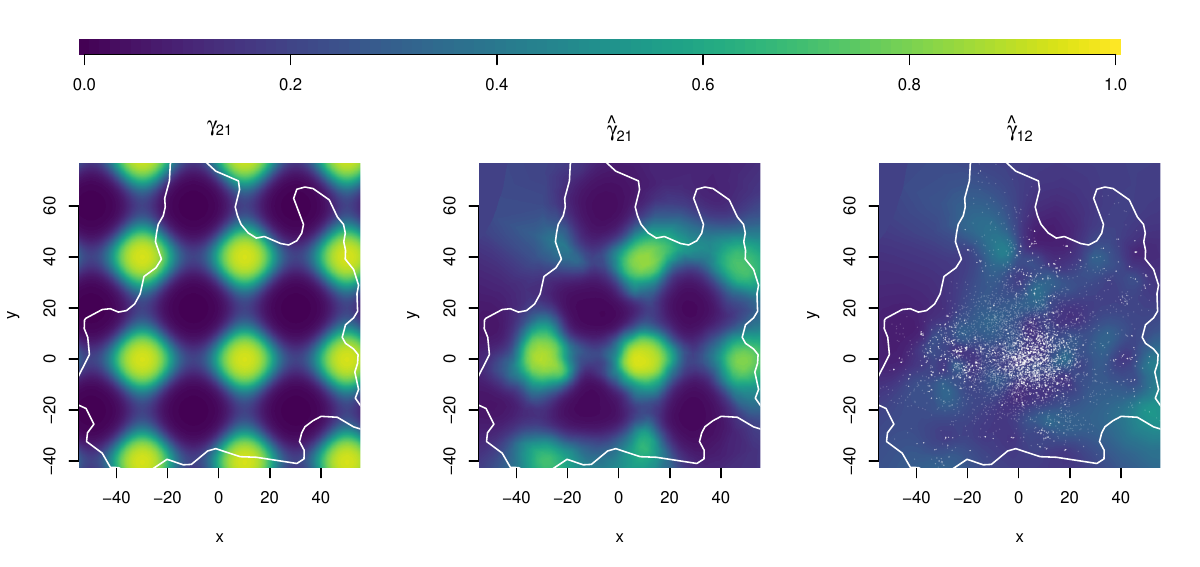}
    \caption{Estimated spatial field in one of the 200 simulation runs with 10,000 observations and bandwidth $k = 15$.
    True transition probability $\gamma_{21}(\bm r)$ as a function of location (left), estimated transition probability $\hat\gamma_{21}(\bm r)$ (middle) and estimated transition probability $\hat\gamma_{12}(\bm r)$ (right). For the latter, the true function is completely flat, hence the ideal fit would be a flat surface. The white lines indicate the non-convex hull around the observed locations used for mesh construction in the GMRF approximation. The right panel additionally shows the visited locations as white dots.}
    \label{fig:sim_spatial_field}
\end{figure}

For time-series length $T \in \{5000, 10000\}$, we simulate 200 data sets each. For each simulation run, we then fit the corresponding HMM with bandwidth parameters $k \in \{2, 5, 10, 15\}$ to investigate the practical effect of varying this parameter. For each run, we construct a GMRF approximation to the spatial field based on a triangulated mesh. However, we misspecify the model by additionally including a spatial field for the transition probability $\gamma_{12}$ in order to inspect the model's capability of smoothing this field appropriately to zero. 
Figure~\ref{fig:sim_spatial_field} shows the true and estimated fields for a single simulation run, indicating adequate recovery of the true field. 

We cannot directly compate the estimated intecept $\beta_0^{(21)}$ to its true value as the predicted Gaussian field $\hat u$ is pulled towards a zero mean by the GMRF density, while the true function $u$ might have a non-zero mean inside the area that was visited. Hence, we compute the mean of $\hat \eta_{21}^{(t)} - \eta_{21}^{(t)}$ over time, and similarly for $\eta_{12}^{(t)}$.
We find a minor upward bias in the average difference for both linear predictors, i.e.\ the estimated models tend to be slightly less persistent than the true one on average for both sample sizes and all bandwidth parameters (cf.\ Figure~\ref{fig:sim_eta}). Overall, the bandwidth parameter seems to have little influence on the estimates: for $k \ge 5$ all distributions are virtually identical.

\begin{figure}
    \centering
    \includegraphics[width=1\linewidth]{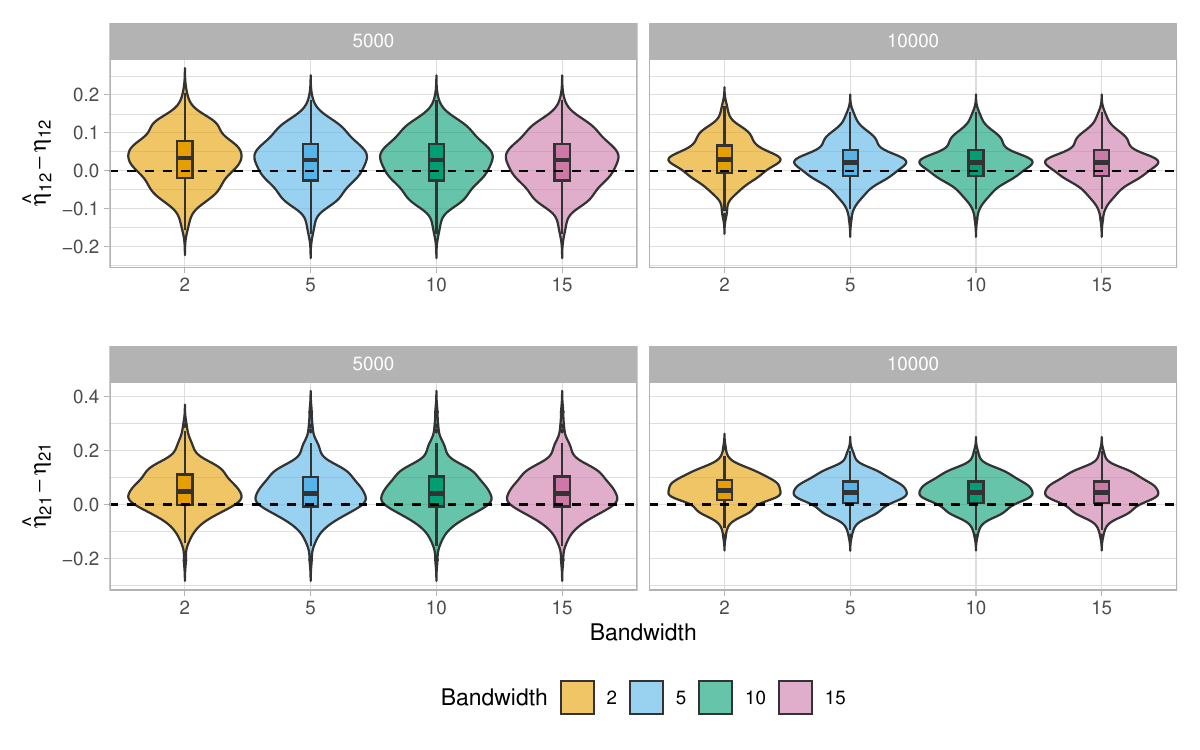}
    \caption{Violin plots and boxplots of the average difference $\tfrac{1}{T} \sum_{t=1}^T (\hat\eta_{ij}^{(t)} - \eta_{ij}^{(t)})$ between estimated and true linear predictors  in sample, i.e.\ evaluated at the observed loactions for sample sizes $T = 5000, 10000$ and bandwidth parameters $k = 2, 5, 10, 15$.}
    \label{fig:sim_eta}
\end{figure}

To investigate interpolation performance, we compute the true and estimated field on a $512 \times 512$ grid in each run. As we do not expect to approximate the true field outside the area of observations, we restrict ourselves to those grid cells that lie inside the non-convex hull spanned by observed locations, which is computed anyway in each run for mesh construction (cf. Figure~\ref{fig:sim_spatial_field}).
Inside this region, the correlation between the estimate $\hat \gamma_{21}(\bm s)$ and its true value $\gamma_{21}(\bm s)$ is computed on the grid. Computing the corresponding corresponding quantity for $\hat \gamma_{12}(\bm s)$ is not possible as the true field is constant and hence has zero variance. 
We find that correlation is high, with an average value slightly below 0.9 (cf.\ Figure~\ref{fig:sim_cor_rmse}). As expected, the distribution becomes more concentrated for the larger sample size of $T = 10000$. Again, virtually no effect of the bandwidth parameter is visible.

\begin{figure}
    \centering
    \includegraphics[width=1\linewidth]{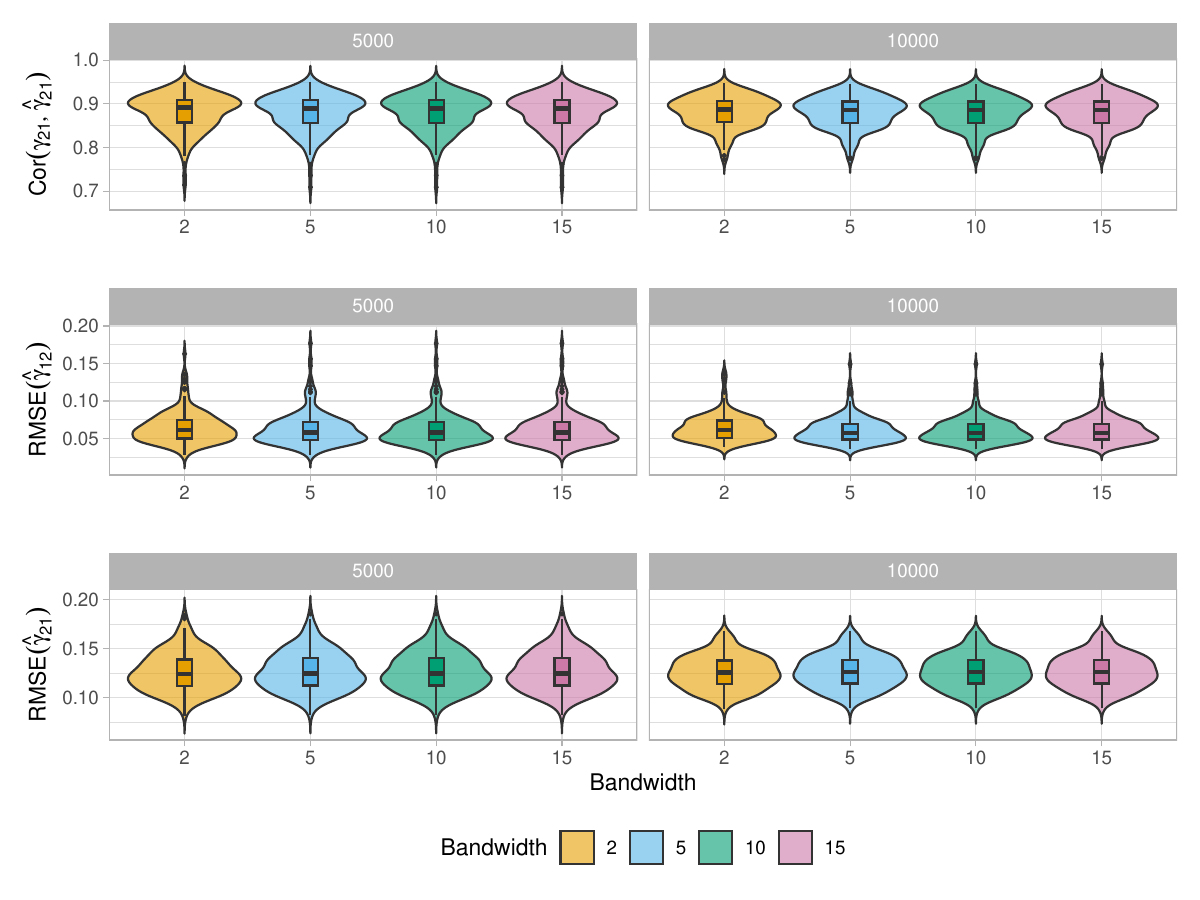}
    \caption{Violin plots and boxplots of correlation between true and estimated transition probability $\gamma_{21}(\bm s)$ on a grid (inside the non-convex hull spanned by the observed locations) (top panel), and root mean squared error (RMSE) between estimated and true transition probabilities (bottom two panels). The correlation between $\hat\gamma_{12}(\bm s)$ and the constant true field is undefined as the latter has zero variance.}
    \label{fig:sim_cor_rmse}
\end{figure}

Additionally, we compute the root mean squared errors (RMSE)
$$
\sqrt{\frac{1}{n} \sum_i \bigl( \hat \gamma_{12}(\bm s_i) - \gamma_{12} \bigr)^2}, \quad \text{and} \quad
\sqrt{\frac{1}{n} \sum_i \bigl( \hat \gamma_{21}(\bm s_i) - \gamma_{21}(\bm s_i) \bigr)^2},$$
where $n$ is the number of grid cells inside the area of interest.
The RMSE between $\hat \gamma_{12}(\bm s)$ and the constant field has an average value around 0.06, with the distribution stabilising for $k \ge 5$ (cf.\ Figure~\ref{fig:sim_cor_rmse}). Performance for the true field is slightly worse, with an average RMSE around 0.12. Again deviation in the distribution of values is only visible between $k = 2$ and $k = 5$.

Lastly, the estimates for the state-dependent means $\mu_1$ and $\mu_2$ and standard deviations $\sigma_1$ and $\sigma_2$ appear mostly unbiased (cf.\ Figure~\ref{fig:sim_mu_sigma}). As is to be expected, variance decreases with increasing sample size. The effect of the bandwidth parameter is barely noticable, again only with visible differences between $k = 2$ and $k = 5$.


\newpage

\subsection{GP with oscillating covariance function}
\label{a4:oscGP}

For the stellar flare case study, we use the Gaussian process detailed in Section 3.3 of \citet{lindgren2011explicit}. This is based on a complex version of the SPDE given in \eqref{eq:spde}, i.e.\ 
$$
\bigl( \kappa^2 \exp(i \pi \omega) - \Delta \bigr) \bigl( u_1(\bm{r}) + i u_2(\bm{r}) \bigr) = \mathcal{W}_1(\bm{r}) + i \mathcal{W}_2(\bm{r})
$$
where $\omega \in (0,1)$ is the oscillation parameter, and $\mathcal{W}_1(\bm{r})$ and $\mathcal{W}_2(\bm{r})$ are independent Gaussian white noise processes.
Oscillation increases with the parameter $\omega$.
The real and imaginary parts, $u_1(\bm{s})$ and $u_2(\bm{s})$, are independent with identical spectra. Hence, the same oscillating covariance function
$$
\mathcal{K}(v, w) = \frac{1}{2 \sin(\pi \omega) \kappa^3} \exp \Bigl( - \kappa \cos \bigl(\frac{\pi \omega}{2} \bigr) \vert v-w \vert \Bigr) \sin \Bigl( \frac{\pi \omega}{2} + \kappa \sin \bigl(\frac{\pi \omega}{2} \bigr) \vert v-w \vert  \Bigr)
$$
applies to either component. Its simple form combines exponential decay with periodic oscillation. Figure~\ref{fig:osc_cov} gives some intuition for the behaviour of $\mathcal{K}$.
The associated precision matrix has closed form and is given in the main text.

\begin{figure}[h]
    \centering
    \includegraphics[width=0.75\linewidth]{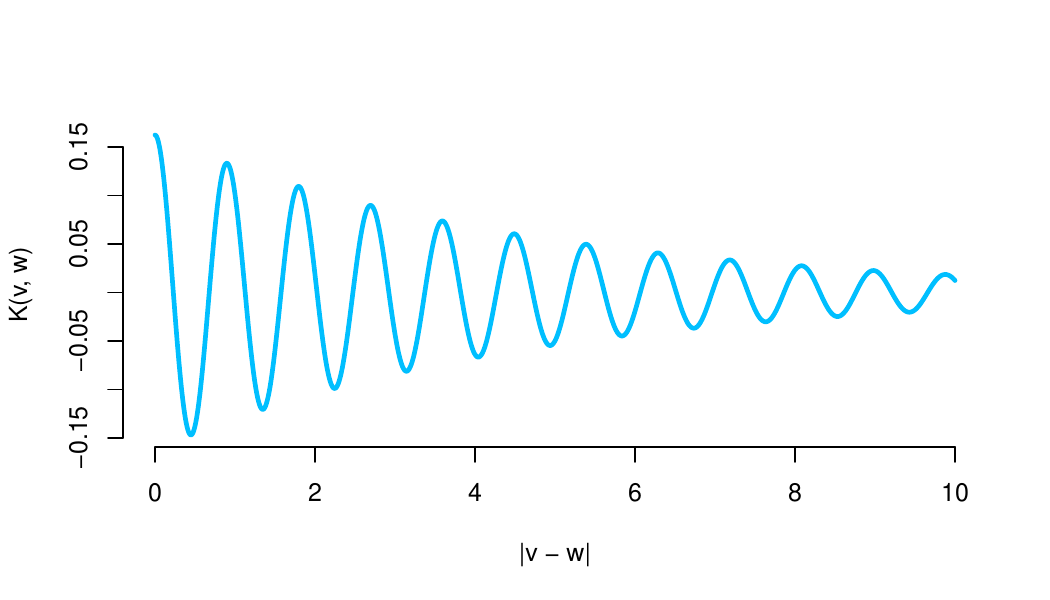}
    \caption{Oscillating covariance function $K(v,w)$ for parameters $\omega = 0.98$ and $\kappa = 7$.}
    \label{fig:osc_cov}
\end{figure}

\newpage
\subsection{Additional figures}
\label{a5:figs}

\begin{figure}[h]
    \centering
    \includegraphics[width=1\linewidth]{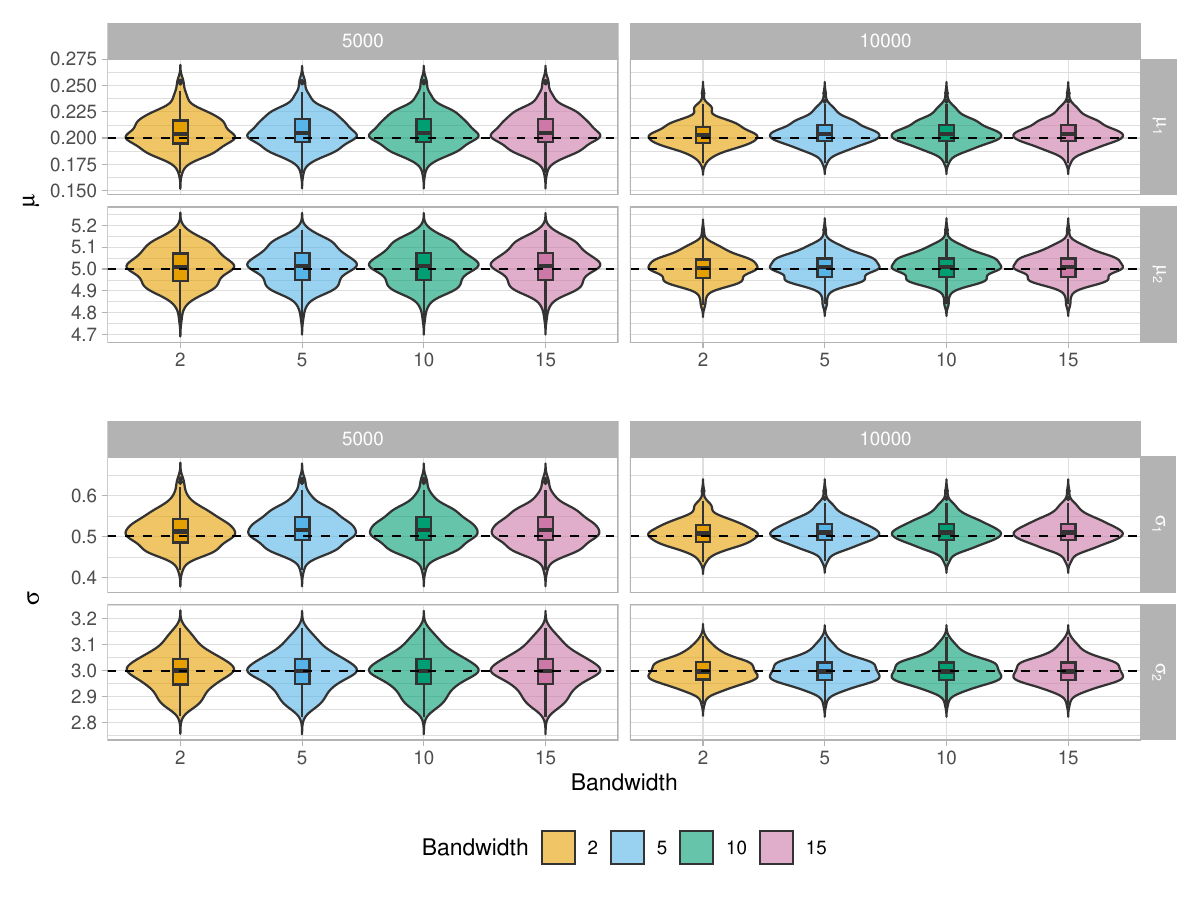}
    \caption{Violin plots and boxplots of the estimated state-dependent means and standard deviations for $T = 5000, 10000$ observations and bandwidth parameter $k = 2, 5, 10, 15$ in the simulation experiment.}
    \label{fig:sim_mu_sigma}
\end{figure}

\begin{figure}[h]
    \centering
    \includegraphics[width=1\linewidth]{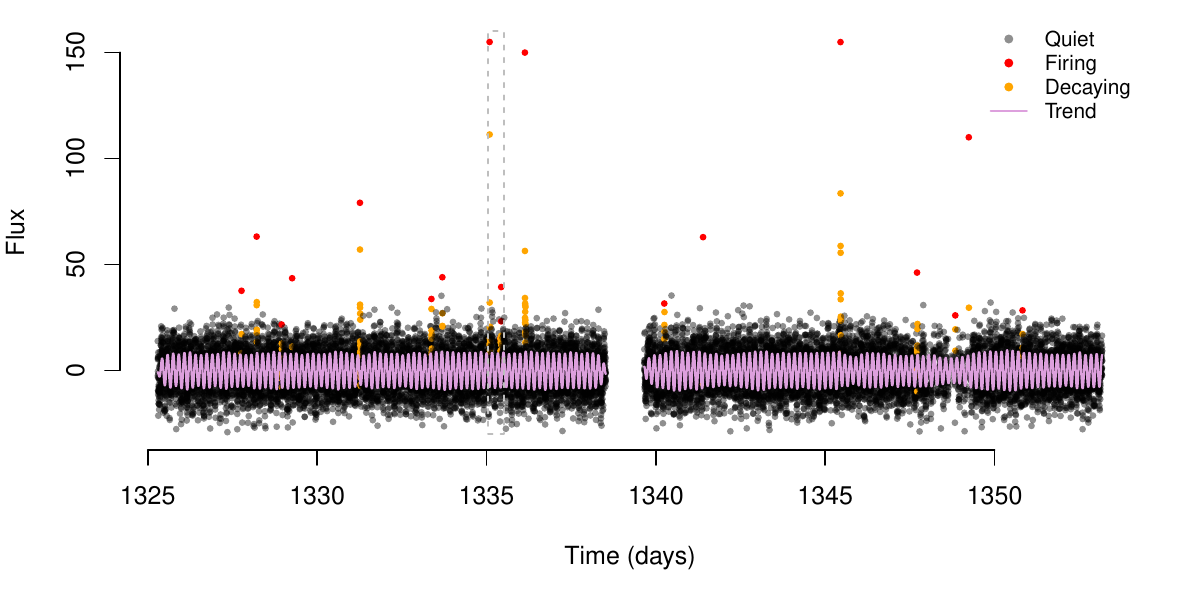}
    \caption{Complete time series of M dwarf TIC 031381302, coloured according to the most likely decoded state.
    The dashed gray box shows the section of the time series depicted in the main manuscript.}
    \label{fig:flares_full}
\end{figure}

\begin{figure}[h]
    \centering
    \includegraphics[width=1\linewidth]{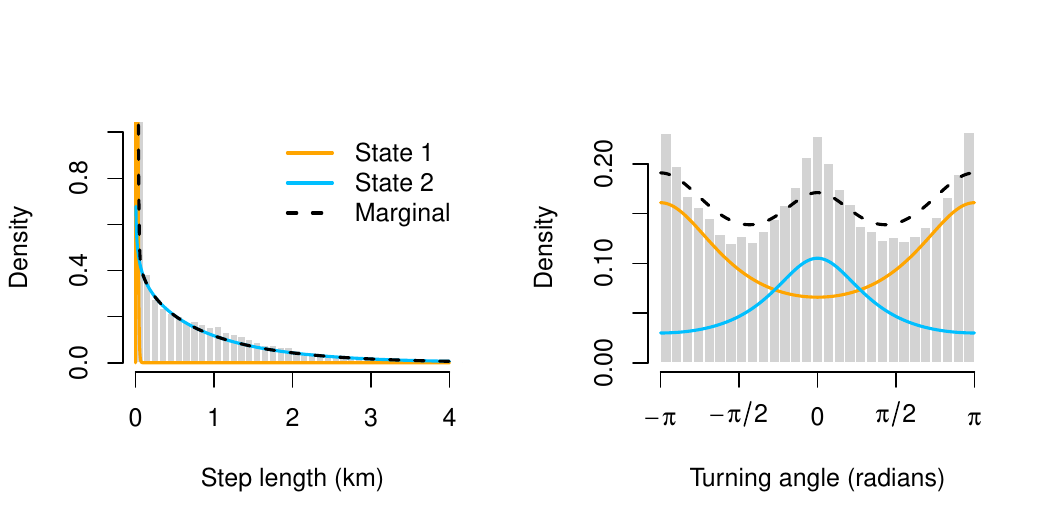}
    \caption{State-dependent distributions for the 2-state lion HMM, weighted by the approximate frequency of state occupancy (based on the decoded states).}
    \label{fig:statedep_lions}
\end{figure}

\begin{figure}[h]
    \centering
    \includegraphics[width=1\linewidth]{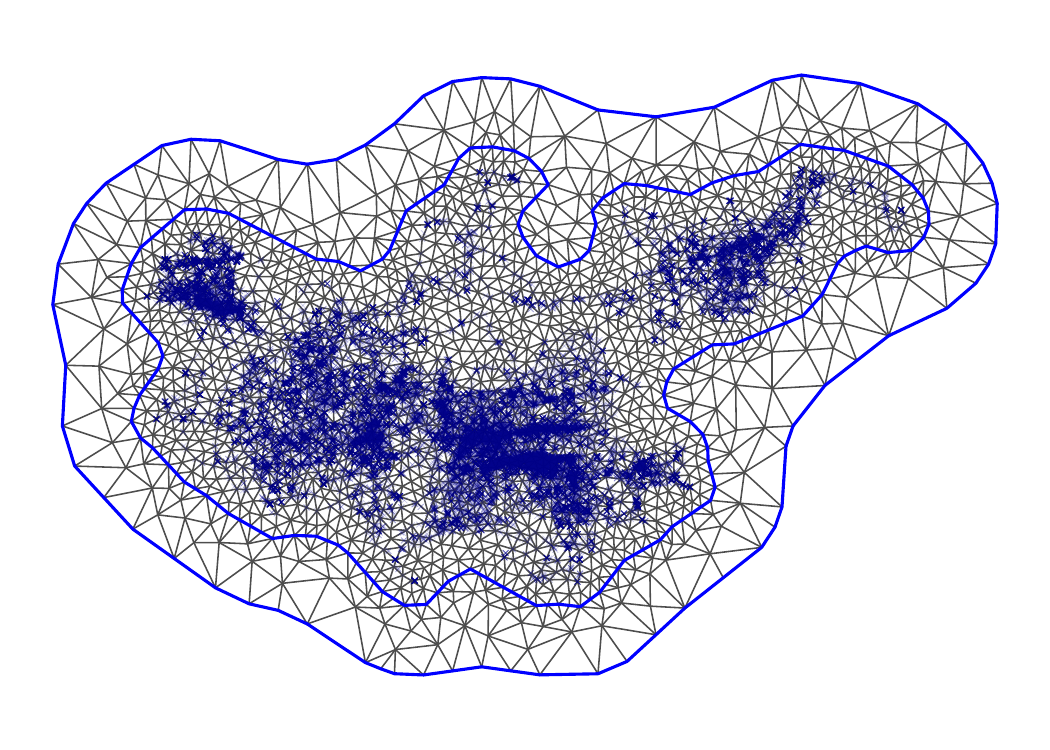}
    \caption{Triangulated mesh used to approximate the Gaussian field in the lion HMM by a GMRF. Observed locations are shown as dark blue crosses.}
    \label{fig:mesh}
\end{figure}

\end{document}